\documentclass[conference,letterpaper]{IEEEtran}%
\usepackage[utf8]{inputenc}
\usepackage[T1]{fontenc}
\usepackage{url}
\usepackage{ifthen}
\usepackage{cite}
\usepackage[cmex10]{amsmath}
\usepackage{amsfonts}
\usepackage{amssymb}
\usepackage{mathtools}
\usepackage{hyperref}
\usepackage{graphicx}
\usepackage{color,xcolor}
\providecommand{\U}[1]{\protect\rule{.1in}{.1in}}
\newtheorem{theorem}{Theorem}

\newtheorem{corollary}[theorem]{Corollary}

\newtheorem{lemma}[theorem]{Lemma}

\newtheorem{proposition}[theorem]{Proposition}

\newenvironment{proof}[1][Proof]{\noindent\textbf{#1.} }{\ \rule{0.5em}{0.5em}}
\allowdisplaybreaks

\newcommand{\Tr}{\operatorname{Tr}}

\begin{document}

\title{Upper Bounds on the Distillable Randomness of Bipartite Quantum States}

\author{%
  \IEEEauthorblockN{Ludovico Lami\IEEEauthorrefmark{1}\IEEEauthorrefmark{2}, Bartosz Regula\IEEEauthorrefmark{3}, Xin Wang\IEEEauthorrefmark{4}, and Mark M.~Wilde\IEEEauthorrefmark{5}}
  \IEEEauthorblockA{
  \IEEEauthorrefmark{1} Institute for Theoretical Physics, Korteweg--de Vries Institute for Mathematics, and QuSoft, \\ University of Amsterdam, the Netherlands. Email: ludovico.lami@gmail.com\\
  \IEEEauthorrefmark{2} Institute for Theoretical Physics and IQST, University of Ulm, Albert-Einstein-Allee 11, D-89069 Ulm, Germany}
  \IEEEauthorblockA{\IEEEauthorrefmark{3}Department of Physics, Graduate School of Science, The University of Tokyo,\\Bunkyo-ku, Tokyo 113-0033, Japan.
                Email: bartosz.regula@gmail.com}
  \IEEEauthorblockA{\IEEEauthorrefmark{4}Institute for Quantum Computing, Baidu Research, Beijing
100193, China.
                    Email: wangxinfelix@gmail.com}
  \IEEEauthorblockA{\IEEEauthorrefmark{5}School of Electrical and Computer
  Engineering, Cornell University, \\ Ithaca, New York 14850, USA.
                    Email: wilde@cornell.edu \vspace{-4ex}}

}

\maketitle

\begin{abstract}
The distillable randomness of a bipartite quantum state is an
information-theoretic quantity equal to the largest net rate at which shared
randomness can be distilled from the state by means of local operations and
classical communication. This quantity has been widely used as a measure of
classical correlations, and one version of it is equal to the regularized Holevo information of the ensemble that
results from measuring one share of the state. However, due to the
regularization, the distillable randomness is difficult to compute in general. To address this problem, we define measures
of classical correlations and prove a number of their properties, most
importantly that they serve as upper bounds on the distillable randomness of an
arbitrary bipartite state. We then further bound these measures from above by some that are efficiently computable by means of
semi-definite programming, we evaluate one of them for the example of an isotropic state, and we
remark on the relation to quantities previously proposed in the literature.
\end{abstract}

%{\it Full version  at \url{https://markwilde.com/RD-bnds.pdf}}

%%%% Single author, or several authors with same affiliation:
%\author{%
%\IEEEauthorblockN{Stefan M.~Moser}
%\IEEEauthorblockA{ETH Zürich\\
%ISI (D-ITET)\\
%CH-8092 Zürich, Switzerland\\
%Email: moser@isi.ee.ethz.ch}
%}

\section{Introduction}

The distillable randomness of a bipartite quantum state is equal to the
largest net rate at which uniformly random, perfectly correlated bits can be
distilled from the state by means of local operations and one-way classical
communication (1W-LOCC)~\cite{DW04}. That is, in this scenario, the net rate
is the rate at which shared randomness is distilled minus the rate at which
classical bits are communicated to accomplish the task. It is important to subtract the classical communication rate, as failure to do so would
trivialize the task, allowing for an infinite number of random bits to be shared. This task fundamentally has its roots in classical information theory
\cite{AC93p1,AC98p2}. Here we also extend the task to allow for general
local operations and classical communication (LOCC).

The distillable randomness of a bipartite state is considered a fundamental
measure of classical correlations contained in a state. There has been a large
literature on this topic, starting with~\cite{HV01,OZ01} and reviewed in~\cite{MBCPV12}, due to the wide
interest in understanding correlations present in quantum states. However,
\cite{DW04} is the one of the few papers to consider understanding these
correlations in an information-theoretic manner (see also
\cite{OHHH02,D05,DHW05RI,KD07,MPZ18,MLA19,CNB22} for other perspectives and related works).

Ref.~\cite{DW04} provided a formal solution for the aforementioned
information-theoretic task. However, like many such formal solutions in
quantum Shannon theory, it does not provide a computationally
efficient procedure for quantifying correlations because it is expressed as a
multi-letter or regularized formula. Our goal here is to fill this void by
providing  upper bounds on the distillable
randomness that are efficiently computable by semi-definite programming. We also justify how some of the newly
proposed measures are themselves measures of classical correlations contained in a
bipartite state.

Although our 
results here can be viewed as a static counterpart to the results in the
dynamical setting from~\cite{WXD18,WFT19,Fang2019a,DKQSWW22}, their derivation requires some new insights, as we explain in detail below. 
Ref.~\cite{DKQSWW22} proved that generalizations of the dynamical measures from~\cite{WXD18,WFT19,Fang2019a} give 
upper bounds on the rate at which classical messages be communicated from one
party to another, whenever they have access to a bipartite channel as a
communication resource. A special case of this setting is when a sender is
communicating classical messages to a receiver, with the assistance of a
classical feedback channel. In contrast, as described earlier, our static
setting here involves two parties distilling shared randomness from a bipartite
state, by means of 1W-LOCC or general LOCC.

We comment here on the novel contributions of our paper when compared to 
earlier works such as~\cite{WXD18,WFT19,Fang2019a,DKQSWW22}. The main measures of classical correlations proposed here, denoted by $\gamma$, $C_\gamma$, and $\boldsymbol{\Gamma}$ in Section~\ref{sec:cl-corr-meas}, are inspired by the measures that have appeared earlier in~\cite{WXD18,WFT19,Fang2019a,DKQSWW22}, denoted there by $\beta$, $C_\beta$, and $\boldsymbol{\Upsilon}$. However, the measures proposed here are symmetric with respect to a swap of the $A$ and $B$ systems. This property is critical for bounding the LOCC-assisted distillable randomness from above, as accomplished in Theorems~\ref{thm:one-shot-bnd} and~\ref{thm:asymp-str-conv-bnd}. Furthermore, we prove here that these measures obey the classical communication bound in Proposition~\ref{prop:classical-comm-bnd}, also known as a ``non-lockability'' property in prior work on correlation measures~\cite{KRBM07,HSKSY21}. Finally, Lemma~\ref{lem:fid-bnd-free-ops} and Proposition~\ref{prop:key-dim-bnd-converse} are critical bounds that allow us to compare the actual state at the end of a randomness distillation protocol with an ideal state, and they ultimately lead to a strong converse upper bound for the distillable randomness. For the task of randomness distillation, the target (ideal) state is a mixed state, which is in distinction to 
most prior works on quantum resource theories, in which the target state is necessarily pure~\cite{liu_2019,FL20,RBTL20}. So the approach we have taken here goes beyond methods previously applied in such contexts and could potentially be useful more generally in research on quantum resource theories.
We note here that similar methods were employed in~\cite{LWD16,Wang2019states}, but Lemma~\ref{lem:fid-bnd-free-ops} and Proposition~\ref{prop:key-dim-bnd-converse} provide a comparison between the ideal mixed state of interest and a convex set of positive semi-definite operators, rather than a comparison between an ideal mixed state and just one other mixed state.

The rest of the paper is organized as follows. In Section~\ref{sec:notation} we introduce some notation and key concepts. 
In Section~\ref{sec:cl-corr-meas} we 
construct and study certain
classical correlation measures, while in 
Section~\ref{sec:dist-rand-state} we define the distillable randomness of a bipartite state. Continuing, in Section~\ref{sec:upper-bnds} we prove that the proposed measures are upper bounds on the distillable randomness, and in Section~\ref{sec:SDP-relax} we consider semi-definite restrictions of the bounds. Finally, in Section~\ref{sec:conclusion} we conclude with a brief summary and some open directions for future research.

\section{Notation} \label{sec:notation}

Here we establish some notation and concepts that we use throughout the paper.
We point the reader to the textbook~\cite{W17} and the manuscript~\cite{KW20book} for
further details. We defer to Appendix~\ref{app:notation} the introduction of the notation required for the proofs presented in the rest of the appendix. %clarification of the notation that we use.
%See Appendix~\ref{app:notation} for further notation needed to understand the proofs in the appendices.
Let $\overline{\Phi}_{AB}^{d}$ denote the maximally classically correlated state of rank $d$, given by
\begin{equation}
\overline{\Phi}_{AB}^{d}\coloneqq\frac{1}{d}\sum_{m=0}^{d-1}|m\rangle\!\langle
m|_{A}\otimes|m\rangle\!\langle m|_{B}. \label{eq:max-class-corr-state}%
\end{equation}
We denote the transpose map acting on the  system$~A$ by
$
T_{A}(\cdot)\coloneqq\sum_{i,j=0}^{d-1}|i\rangle\!\langle j|_{A}%
(\cdot)|i\rangle\!\langle j|_{A}$.

Let us define a generalized divergence $\boldsymbol{D}$ of a state $\rho$ and
a positive semi-definite operator $\sigma$ as a function that obeys:
\begin{enumerate}
\item data processing: $\boldsymbol{D}(\rho\Vert\sigma)\geq\boldsymbol{D}(\mathcal{N}(\rho)\Vert\mathcal{N}(\sigma))$, where $\mathcal{N}$ is an arbitrary quantum channel;

\item the scaling relation: $\boldsymbol{D}(\rho\Vert c\sigma)=\boldsymbol{D}(\rho
\Vert\sigma)-\log_{2}c$, for all $c > 0$; and

\item 
the zero-value property: $\boldsymbol{D}(\rho\Vert\rho)=0$ for every state $\rho$.
\end{enumerate}

We note that the scaling and zero-value properties together imply non-negativity: $\forall c\in(0,1], \ \boldsymbol{D}(1\Vert c)\geq0 $. Indeed, considering that the number 1 is a state of a trivial system, we have that $\boldsymbol{D}(1\Vert c) = \boldsymbol{D}(1\Vert 1) - \log_2 c = - \log_2 c \geq 0$.

\section{Classical Correlation Measures}

\label{sec:cl-corr-meas}

We now define some classical correlation measures that are used later on, in
the application of bounding the distillable randomness from above. For a
positive semi-definite, bipartite operator $\sigma_{AB}$, let us define%
\begin{equation}
\gamma(\sigma_{AB})\coloneqq\inf_{K_{A},L_{B},V_{AB}\in\text{Herm}}\left\{
\begin{array}
[c]{c}%
\operatorname{Tr}[K_{A}\otimes L_{B}]:\\
T_{B}(V_{AB}\pm\sigma_{AB})\geq0,\\
K_{A}\otimes L_{B}\pm V_{AB}\geq0
\end{array}
\right\}  . \label{eq:basic-gamma-measure}%
\end{equation}
We also denote this quantity by $\gamma(A;B)_{\sigma}$. We also define%
\begin{align}
C_{\gamma}(\sigma_{AB})  &  \coloneqq C_{\gamma}(A;B)_{\sigma}\coloneqq\log
_{2}\gamma(\sigma_{AB}).
\end{align}
For a bipartite state $\rho_{AB}$, we then define
\begin{align}
\boldsymbol{\Gamma}(\rho_{AB})  &  \coloneqq\inf_{\substack{\sigma_{AB}%
\geq0:\\\gamma(\sigma_{AB})\leq1}}\boldsymbol{D}(\rho_{AB}\Vert\sigma
_{AB}),\label{eq:Gamma-quantity-gen-div-def}
\end{align}
which follows an approach for defining resource measures proposed originally by~\cite{R99,R01,AdMVW02} in the context of entanglement theory. 
Also, as a consequence of the scaling and zero-value properties listed above,
we conclude that%
\begin{equation}
\boldsymbol{\Gamma}(\rho_{AB})\leq\boldsymbol{D}(\rho_{AB}\Vert\rho
_{AB}/\gamma(\rho_{AB}))=C_{\gamma}(\rho_{AB}).
\label{eq:Big-gamma-C-gamma-relation}%
\end{equation}

Our goal is to show that $\boldsymbol{\Gamma}(\rho_{AB})$ is an upper bound on
the distillable randomness of a bipartite state $\rho_{AB}$. We do so in
Section~\ref{sec:upper-bnds} after proving various properties of
$\gamma$, $C_{\gamma}$, and $\boldsymbol{\Gamma}$, in
Section~\ref{sec:props-cl-corr-measures}.

\subsection{Properties of Classical Correlation Measures}

\label{sec:props-cl-corr-measures}In this section, we establish a number of
properties of the quantities $\gamma$, $C_{\gamma}$, and $\boldsymbol{\Gamma}%
$, proposed in Section~\ref{sec:cl-corr-meas}. In particular, we prove that
these measures satisfy several properties expected for a measure of classical
correlations of a bipartite state, including

\begin{enumerate}

\item symmetry under exchange of $A$ and $B$ (Proposition~\ref{prop:symm-exch-A-B}), 

\item data processing under local channels
(Proposition~\ref{prop:DP-local-chs}),

\item local isometric invariance (Corollary~\ref{cor:local-iso-inv}),

\item non-negativity, 
faithfulness for product states
(Proposition~\ref{prop:non-neg-state-meas}),

\item dimension bound (Proposition~\ref{prop:dim-bound-state-measure}),

\item scale invariance (Proposition~\ref{prop:scale-invariance}),

\item classical communication bound (Proposition~\ref{prop:classical-comm-bnd}),

\item subadditivity (Proposition~\ref{prop:subadd-state-meas}), and 

\item continuity near maximally classically correlated states (Proposition~\ref{prop:key-dim-bnd-converse}).
\end{enumerate}

\begin{proposition}
[Exchange Symmetry]\label{prop:symm-exch-A-B} Let
$\sigma_{AB}$ be a bipartite positive semi-definite operator, and let $\rho_{AB}$ be a state. Then
\begin{align}
C_{\gamma}(A;B)_{\sigma}  & = C_{\gamma}(B;A)_{\sigma}, \\
\boldsymbol{\Gamma}(A;B)_{\rho}  & = \boldsymbol{\Gamma}(B;A)_{\rho}.
\end{align}
\end{proposition}

\begin{proof}
This follows by inspecting the definitions and using the fact that, for a Hermitian operator $W_{AB}$, the inequality $T_B(W_{AB}) \geq 0 $ holds if and only if $T_A(W_{AB}) \geq 0 $ does.
\end{proof}

\begin{proposition}
[Data-Processing under Local Channels]\label{prop:DP-local-chs} Let
$\sigma_{AB}$ be a bipartite positive semi-definite operator, and let
$\mathcal{N}_{A\rightarrow A^{\prime}}$ and $\mathcal{M}_{B\rightarrow
B^{\prime}}$ be quantum channels. Then%
\begin{align}
C_{\gamma}(A;B)_{\sigma}  &  \geq C_{\gamma}(A^{\prime};B^{\prime})_{\omega
},\label{eq:C-gamma-DP-local-chs}\\
\boldsymbol{\Gamma}(A;B)_{\sigma}  &  \geq\boldsymbol{\Gamma}(A^{\prime
};B^{\prime})_{\omega}, \label{eq:bold-gamma-DP-local-chs}%
\end{align}
where $\omega_{A^{\prime}B^{\prime}}\coloneqq(\mathcal{N}_{A\rightarrow
A^{\prime}}\otimes\mathcal{M}_{B\rightarrow B^{\prime}})(\rho_{AB})$.
\end{proposition}

\begin{proof}
See Appendix~\ref{app:data-proc-local-chs}.
\end{proof}

\begin{corollary}
[Local Isometric Invariance]\label{cor:local-iso-inv}Let $\sigma_{AB}$ be a
bipartite positive semi-definite operator, $\mathcal{U}_{A\rightarrow
A^{\prime}}$ an isometric channel, and $\mathcal{V}_{B\rightarrow B^{\prime}}$
an isometric channel. Then%
\begin{align}
C_{\gamma}(A;B)_{\sigma}  &  =C_{\gamma}(A^{\prime};B^{\prime})_{\omega},\\
\boldsymbol{\Gamma}(A;B)_{\sigma}  &  =\boldsymbol{\Gamma}(A^{\prime
};B^{\prime})_{\omega}, \label{eq:Gamma-iso-inv}%
\end{align}
where $\omega_{A^{\prime}B^{\prime}}\coloneqq(\mathcal{U}_{A\rightarrow
A^{\prime}}\otimes\mathcal{V}_{B\rightarrow B^{\prime}})(\sigma_{AB})$.
\end{corollary}

\begin{proof}
The inequality $C_{\gamma}(A;B)_{\sigma}\geq C_{\gamma}(A^{\prime};B^{\prime
})_{\omega}$ follows by a direct application of
Proposition~\ref{prop:DP-local-chs}. Let us define the channels%
\begin{align}
\mathcal{L}_{A^{\prime}\rightarrow A}(\cdot)  &  \coloneqq(\mathcal{U}%
_{A\rightarrow A^{\prime}})^{\dag}(\cdot)+\operatorname{Tr}[(\operatorname{id}%
_{A^{\prime}}-(\mathcal{U}_{A\rightarrow A^{\prime}})^{\dag})(\cdot)]\tau
_{A},\notag \\
\mathcal{M}_{B^{\prime}\rightarrow B}(\cdot)  &  \coloneqq(\mathcal{V}%
_{B\rightarrow B^{\prime}})^{\dag}(\cdot)+\operatorname{Tr}[(\operatorname{id}%
_{B^{\prime}}-(\mathcal{V}_{B\rightarrow B^{\prime}})^{\dag})(\cdot)]\eta_{B},
\end{align}
where $\tau_{A}$ and $\eta_{B}$ are arbitrary quantum states. Consider that
$\mathcal{L}_{A^{\prime}\rightarrow A}\circ\mathcal{U}_{A\rightarrow
A^{\prime}}=\operatorname{id}_{A}$ and $\mathcal{M}_{B^{\prime}\rightarrow
B}\circ\mathcal{V}_{B\rightarrow B^{\prime}}=\operatorname{id}_{B}$. Then the
opposite inequality also follows by applying
Proposition~\ref{prop:DP-local-chs}:%
\begin{align}
C_{\gamma}(A^{\prime};B^{\prime})_{\omega}  &  \geq C_{\gamma}(A;B)_{\zeta}
  =C_{\gamma}(A;B)_{\sigma},
\end{align}
where $\zeta_{AB}\coloneqq(\mathcal{L}_{A^{\prime}\rightarrow A}%
\otimes\mathcal{M}_{B^{\prime}\rightarrow B})(\omega_{A^{\prime}B^{\prime}%
})=\sigma_{AB}$. The same line of reasoning establishes~\eqref{eq:Gamma-iso-inv}.
\end{proof}

\begin{proposition}
[Non-Negativity and Faithfulness]\label{prop:non-neg-state-meas}Let $\rho_{AB}$ be a bipartite
state. Then $C_{\gamma}(A;B)_{\rho}$ and $\boldsymbol{\Gamma}(A;B)_{\rho}%
$\ are non-negative, i.e.
\begin{equation}
C_{\gamma}(A;B)_{\rho}\geq0,\qquad\boldsymbol{\Gamma}(A;B)_{\rho}\geq0 ,
\end{equation}
and equal to zero if $\rho_{AB}=\sigma_{A}\otimes \tau_{B}$ is a product state. Furthermore, $C_{\gamma}$ is faithful, i.e., $C_\gamma(\rho_{AB})=0$ implies that $\rho_{AB}$ is a product state, and $\boldsymbol{\Gamma}$ is faithful provided that the underlying divergence $\boldsymbol{D}$ is itself faithful (positive definite) and also lower semicontinuous in its second argument.
\end{proposition}

\begin{proof}
See Appendix~\ref{app:non-neg}.
\end{proof}

\begin{proposition}
[Dimension Bound]\label{prop:dim-bound-state-measure}Let $\rho_{AB}$ be a
bipartite state. Then
\begin{equation}
\boldsymbol{\Gamma}(A;B)_{\rho}\leq C_{\gamma}(A;B)_{\rho}\leq\log_{2}
\min\left\{  d_{A},d_{B}\right\} ,  \label{eq:dim-bnd}
\end{equation}
where $d_A$ and $d_B$ denote the dimensions of the local systems.
\end{proposition}

\begin{proof}
See Appendix~\ref{app:dim-bnd}.
\end{proof}

\begin{proposition}
[Scale Invariance]\label{prop:scale-invariance}Let $\sigma_{AB}$ be a positive
semi-definite, bipartite operator, and let $c>0$. Then%
\begin{equation}
\gamma(c\sigma_{AB})=c\gamma(\sigma_{AB}).
\end{equation}

\end{proposition}

\begin{proof}
Let $K_{A}$, $L_{B}$, and $V_{AB}$ be arbitrary choices for the optimization
problem for $\gamma(\sigma_{AB})$. Then $cK_{A}$, $L_{B}$, and $cV_{AB}$ are
particular choices for which the objective function evaluates to
$\operatorname{Tr}[cK_{A}\otimes L_{B}]=c\operatorname{Tr}[K_{A}\otimes
L_{B}]$ and such that the constraints for $c\sigma_{AB}$ hold. It follows then
that
$
\gamma(c\sigma_{AB})\leq c\gamma(\sigma_{AB})$.
To see the opposite inequality, consider applying this again to find that
$
\gamma(c^{-1}c\sigma_{AB})\leq c^{-1}\gamma(c\sigma_{AB})$,
which is equivalent to
$
c\gamma(\sigma_{AB})\leq\gamma(c\sigma_{AB})$.
\end{proof}

\begin{proposition}
[Classical Communication Bound]\label{prop:classical-comm-bnd}Let $\rho_{XAB}$
be a tripartite state, for which system $X$ is classical, i.e.,%
\begin{equation}
\rho_{XAB}\coloneqq \sum_{x}p(x)|x\rangle\!\langle x|_{X}\otimes\rho_{AB}^{x},
\end{equation}
where $\{p(x)\}_{x}$ is a probability distribution and $\{\rho_{AB}^{x}\}_{x}$
is a set of states. Then%
\begin{align}
C_{\gamma}(AX;B)_{\rho}  &  \leq\log_{2}d_{X}+C_{\gamma}(A;BX)_{\rho
},\label{eq:c_gam-comm-bnd}\\
\boldsymbol{\Gamma}(AX;B)_{\rho}  &  \leq\log_{2}d_{X}+\boldsymbol{\Gamma
}(A;BX)_{\rho}. \label{eq:big_gam-comm-bnd}%
\end{align}

\end{proposition}

\begin{proof}
See Appendix~\ref{app:classical-comm-bnd}.
\end{proof}

\begin{proposition}
[Subadditivity]\label{prop:subadd-state-meas}Let $\rho_{A_{1}A_{2}B_{1}B_{2}%
}\coloneqq\sigma_{A_{1}B_{1}}\otimes\tau_{A_{2}B_{2}}$, where $\sigma
_{A_{1}B_{1}}$ and $\tau_{A_{2}B_{2}}$ are bipartite operators. Then the
following subadditivity inequality holds%
\begin{equation}
C_{\gamma}(A_{1}A_{2};B_{1}B_{2})_{\rho}\leq C_{\gamma}(A_{1};B_{1})_{\sigma
}+C_{\gamma}(A_{2};B_{2})_{\tau}. \label{eq:subadd-c-gamma}%
\end{equation}
If the underlying generalized divergence $\boldsymbol{D}$ is subadditive, then
$\boldsymbol{\Gamma}$ is subadditive for a state $\rho_{A_{1}A_{2}B_{1}B_{2}%
}=\sigma_{A_{1}B_{1}}\otimes\tau_{A_{2}B_{2}}$:
\begin{equation}
\boldsymbol{\Gamma}(A_{1}A_{2};B_{1}B_{2})_{\rho}\leq\boldsymbol{\Gamma}%
(A_{1};B_{1})_{\sigma}+\boldsymbol{\Gamma}(A_{2};B_{2})_{\tau}.
\label{eq:subadd-big-gamma}%
\end{equation}

\end{proposition}

\begin{proof}
See Appendix~\ref{app:subadd-state-meas}.
\end{proof}

\begin{lemma}
\label{lem:fid-bnd-free-ops}The following bound holds:%
\begin{equation}
\sup_{\sigma_{AB}\geq0 :  \gamma(\sigma_{AB})\leq1} F\Big(\overline{\Phi
}_{AB}^{d},\sigma_{AB}\Big)\leq\frac{1}{d},
\end{equation}
where $\overline{\Phi}_{AB}^{d}$ is the maximally classical correlated state
from~\eqref{eq:max-class-corr-state} and $F(\omega,\tau)\coloneqq \left\Vert
\sqrt{\omega}\sqrt{\tau}\right\Vert _{1}^{2}$ is the fidelity of states
$\omega$ and$~\tau$~\cite{U76}.
\end{lemma}

\begin{proof}
See Appendix~\ref{app:pf-fid-bnd-free-ops}.
\end{proof}

Recall that the sandwiched R\'{e}nyi relative entropy of a state~$\rho$ and a
positive semi-definite operator $\sigma$ is defined for $\alpha\in\left(
0,1\right)  \cup(1,\infty)$ as~\cite{MDSFT13,WWY13}%
\begin{equation}
\widetilde{D}_{\alpha}(\rho\Vert\sigma)\coloneqq\frac{2\alpha}{\alpha-1}%
\log_{2}\left\Vert \sigma^{\left(  1-\alpha\right)  /2\alpha}\rho
^{1/2}\right\Vert _{2\alpha}.\label{eq:sandwich-renyi-def}%
\end{equation}
It converges to the quantum relative entropy~\cite{U62}%
\begin{equation}
D(\rho\Vert\sigma)\coloneqq\operatorname{Tr}[\rho(\log_{2}\rho-\log_{2}%
\sigma)]\label{eq:q-rel-entr}%
\end{equation}
in the limit $\alpha\rightarrow1$, and it is a generalized divergence for $\alpha \geq 1/2$. See~\cite{KW20book} for further properties.

\begin{proposition}
\label{prop:key-dim-bnd-converse}Let $\omega_{AB}$ be a state satisfying%
\begin{equation}
F\Big(\overline{\Phi}_{AB}^{d},\omega_{AB}\Big)\geq1-\varepsilon,
\end{equation}
for $\varepsilon\in\left[  0,1\right]  $. Then, for all $\alpha>1$,%
\begin{equation}
\log_{2}d\leq\widetilde{\Gamma}_{\alpha}(\omega_{AB})+\frac{\alpha}{\alpha
-1}\log_{2}\!\left(  \frac{1}{1-\varepsilon}\right)  ,
\end{equation}
where $\widetilde{\Gamma}_{\alpha}(\omega_{AB})$ is defined from
\eqref{eq:Gamma-quantity-gen-div-def} with the underlying divergence taken to
be the sandwiched R\'{e}nyi relative entropy~$\widetilde{D}_{\alpha}$.
\end{proposition}

\begin{proof}
See Appendix~\ref{app:pseudo-cont-bnd}.
\end{proof}

\section{Distillable Randomness of a Bipartite State}

\label{sec:dist-rand-state}

Let $\rho_{AB}$ be a bipartite state. A protocol for randomness distillation
assisted by 1W-LOCC begins with a quantum channel of the form $\mathcal{E}%
_{A\rightarrow LM}$, where the output systems are classical. The system $L$ is
communicated to Bob over a noiseless classical channel. After that, he acts
with a quantum channel $\mathcal{D}_{BL\rightarrow M^{\prime}}$. The state at
the end of the protocol is thus%
\begin{equation}
\omega_{MM^{\prime}}\coloneqq(\mathcal{D}_{BL\rightarrow M^{\prime}}%
\circ\mathcal{E}_{A\rightarrow LM})(\rho_{AB}).
\end{equation}
A $(d,\varepsilon)$ protocol for randomness distillation satisfies%
\begin{equation}
p_{\text{err}}((\mathcal{E},\mathcal{D});\rho_{AB})\coloneqq1-F(\omega
_{MM^{\prime}},\overline{\Phi}_{MM^{\prime}}^{d})\leq\varepsilon
.\label{eq:err-criterion-RD}%
\end{equation}
The one-shot distillable randomness of $\rho_{AB}$ is defined as%
\begin{multline}
R^{\varepsilon}(\rho_{AB})\coloneqq\\
\sup_{\substack{\mathcal{E}_{A\rightarrow LM},\\\mathcal{D}_{BL\rightarrow
M^{\prime}}}}\left\{  \log_{2}d-\log_{2}d_{L}:p_{\text{err}}((\mathcal{E}%
,\mathcal{D});\rho_{AB})\leq\varepsilon\right\}  .
\end{multline}
In words, it is the largest net rate at which maximally correlated randomness
can be generated. Indeed, we need to subtract off the classical communication
used in the protocol; 
otherwise, the quantity would be trivially equal to $+\infty$.

The distillable randomness of $\rho_{AB}$ is defined as%
\begin{equation}
R(\rho_{AB})\coloneqq\inf_{\varepsilon\in(0,1]}\liminf_{n\rightarrow\infty
}\frac{1}{n}R^{\varepsilon}(\rho_{AB}^{\otimes n}),
\end{equation}
and the strong converse distillable randomness as%
\begin{equation}
\widetilde{R}(\rho_{AB})\coloneqq\sup_{\varepsilon\in\lbrack0,1)}%
\limsup_{n\rightarrow\infty}\frac{1}{n}R^{\varepsilon}(\rho_{AB}^{\otimes n}).
\end{equation}

We can generalize these definitions to allow for assistance by arbitrary LOCC
(see~\cite{CLM+14}\ for a detailed account of LOCC). A protocol for randomness
distillation assisted by LOCC\ begins with Alice performing the channel
$\mathcal{E}_{A\rightarrow L_{1}A_{1}}^{(1)}$, where the system $L_{1}$ is
classical and communicated to Bob. Bob then performs the channel
$\mathcal{D}_{L_{1}B\rightarrow L_{2}B_{2}}^{(2)}$, where the system $L_{2}$
is classical and communicated to Alice. This process continues for $k$ rounds; we denote Alice's other channels by $\{\mathcal{E}_{L_{i-1}%
A_{i-2}\rightarrow L_{i}A_{i}}^{(i)}\}_{i}$, for $i\in\left\{  3,5,\ldots
\right\}$, and Bob's other channels by $\{\mathcal{D}_{L_{i-1}B_{i-2}%
\rightarrow L_{i}B_{i}}^{(i)}\}_{i}$, for $i\in\left\{  4,6,\ldots\right\}$.
The last two channels in the protocol, without loss of generality, are then
$\mathcal{E}_{L_{k-1}A_{k-2}\rightarrow L_{k}M}^{(k)}$ for Alice and
$\mathcal{D}_{L_{k}B_{k-1}\rightarrow M^{\prime}}^{(k+1)}$ for Bob. The state
at the end of the protocol is thus%
\begin{equation}
\omega_{MM^{\prime}}\coloneqq(\mathcal{D}^{(k+1)}\circ\mathcal{E}^{(k)}%
\circ\cdots\circ\mathcal{D}^{(2)}\circ\mathcal{E}^{(1)})(\rho_{AB}%
).\label{eq:final-state-k-round-prot}%
\end{equation}
A $(d,\varepsilon)$ protocol for LOCC-assisted randomness distillation 
 satisfies
\begin{equation}
p_{\text{err}}(\mathcal{P}^{(k)};\rho_{AB})\coloneqq1-F(\omega_{MM^{\prime}%
},\overline{\Phi}_{MM^{\prime}}^{d})\leq\varepsilon,
\end{equation}
where $\mathcal{P}^{(k)}$ is a shorthand for the whole protocol, i.e.,
$\mathcal{P}^{(k)}=\{\mathcal{E}^{(1)},\mathcal{D}^{(2)},\ldots,\mathcal{E}%
^{(k)},\mathcal{D}^{(k+1)}\}$. The one-shot distillable randomness of
$\rho_{AB}$, assisted by LOCC, is defined as%
\begin{multline}
R_{\leftrightarrow}^{\varepsilon}(\rho_{AB})\coloneqq\\
\sup_{k\in\mathbb{N}}\sup_{\mathcal{P}^{(k)}}\left\{  \log_{2}d-\sum_{i=1}%
^{k}\log_{2}d_{L_{i}}:p_{\text{err}}(\mathcal{P}^{(k)};\rho_{AB}%
)\leq\varepsilon\right\}  .
\end{multline}
The LOCC-assisted distillable randomness of a bipartite state $\rho_{AB}$ is defined as
\begin{equation}
R_{\leftrightarrow}(\rho_{AB})\coloneqq\inf_{\varepsilon\in(0,1]}%
\liminf_{n\rightarrow\infty}\frac{1}{n}R_{\leftrightarrow}^{\varepsilon}%
(\rho_{AB}^{\otimes n}),
\end{equation}
and the strong converse LOCC-assisted distillable randomness as%
\begin{equation}
\widetilde{R}_{\leftrightarrow}(\rho_{AB})\coloneqq\sup_{\varepsilon\in
\lbrack0,1)}\limsup_{n\rightarrow\infty}\frac{1}{n}R_{\leftrightarrow
}^{\varepsilon}(\rho_{AB}^{\otimes n}).
\end{equation}

The following inequalities hold by definition:
\begin{align}
R^{\varepsilon}(\rho_{AB})  & \leq R_{\leftrightarrow}^{\varepsilon}(\rho
_{AB}),\qquad R(\rho_{AB})\leq R_{\leftrightarrow}(\rho_{AB}),\\
\widetilde{R}(\rho_{AB})  & \leq\widetilde{R}_{\leftrightarrow}(\rho
_{AB}),\qquad R(\rho_{AB})\leq\widetilde{R}(\rho_{AB}),\\
R_{\leftrightarrow}(\rho_{AB})  & \leq\widetilde{R}_{\leftrightarrow}%
(\rho_{AB}).
\end{align}
We provide upper bounds for $R_{\leftrightarrow}^{\varepsilon}(\rho_{AB})$ and
$\widetilde{R}_{\leftrightarrow}(\rho_{AB})$ in the next section.

\section{Upper Bounds on the Distillable Randomness}

\label{sec:upper-bnds}

\begin{theorem}
\label{thm:one-shot-bnd}The following bound holds for all $\alpha>1$:%
\begin{equation}
R_{\leftrightarrow}^{\varepsilon}(\rho_{AB})\leq\widetilde{\Gamma}_{\alpha
}(\rho_{AB})+\frac{\alpha}{\alpha-1}\log_{2}\!\left(  \frac{1}{1-\varepsilon
}\right)  ,\label{eq:LOCC-assisted-up-bnd}%
\end{equation}
where $\widetilde{\Gamma}_{\alpha}(\rho_{AB})$ is defined from
\eqref{eq:Gamma-quantity-gen-div-def} and~\eqref{eq:sandwich-renyi-def}.
\end{theorem}

\begin{proof}
Let us begin by proving the bound%
\begin{equation}
R^{\varepsilon}(\rho_{AB})\leq\widetilde{\Gamma}_{\alpha}(\rho_{AB}%
)+\frac{\alpha}{\alpha-1}\log_{2}\!\left(  \frac{1}{1-\varepsilon}\right)  ,
\end{equation}
and then we discuss how to generalize the proof afterward. Consider an
arbitrary protocol for randomness distillation assisted by 1W-LOCC. Since the
final state $\omega_{MM^{\prime}}$ satisfies~\eqref{eq:err-criterion-RD}, we
apply Proposition~\ref{prop:key-dim-bnd-converse} to conclude that%
\begin{equation}
\log_{2}d\leq\widetilde{\Gamma}_{\alpha}(M;M^{\prime})_{\omega}+\frac{\alpha
}{\alpha-1}\log_{2}\!\left(  \frac{1}{1-\varepsilon}\right)  .
\end{equation}
Next we apply data processing under local channels
(Proposition~\ref{prop:DP-local-chs}) to conclude that%
\begin{equation}
\widetilde{\Gamma}_{\alpha}(M;M^{\prime})_{\omega}\leq\widetilde{\Gamma
}_{\alpha}(M;BL)_{\mathcal{E}(\rho)}.
\end{equation}
We then apply Propositions~\ref{prop:symm-exch-A-B} and~\ref{prop:classical-comm-bnd} to find that
\begin{equation}
\widetilde{\Gamma}_{\alpha}(M;BL)_{\mathcal{E}(\rho)}\leq\log_{2}%
d_{L}+\widetilde{\Gamma}_{\alpha}(LM;B)_{\mathcal{E}(\rho)}.
\end{equation}
Next we apply data processing under local channels again
(Proposition~\ref{prop:DP-local-chs}) to conclude that%
\begin{equation}
\widetilde{\Gamma}_{\alpha}(LM;B)_{\mathcal{E}(\rho)}\leq\widetilde{\Gamma
}_{\alpha}(A;B)_{\rho}.
\end{equation}
Putting everything together, we finally conclude that%
\begin{equation}
\log_{2}d-\log_{2}d_{L}\leq\widetilde{\Gamma}_{\alpha}(A;B)_{\rho}%
+\frac{\alpha}{\alpha-1}\log_{2}\!\left(  \frac{1}{1-\varepsilon}\right)  .
\end{equation}
Since this holds for an arbitrary randomness distillation protocol, we conclude the desired bound.

To extend this to $R_{\leftrightarrow}^{\varepsilon}(\rho_{AB})$, we can
iterate the same reasoning, going backward through a protocol of the form
discussed around~\eqref{eq:final-state-k-round-prot}, using
Propositions~\ref{prop:symm-exch-A-B},~\ref{prop:DP-local-chs}, and~\ref{prop:classical-comm-bnd} to
arrive at the following upper bound for a $k$-round protocol:%
\begin{equation}
\log_{2}d-\sum_{i=1}^{k}\log_{2}d_{L_{i}}\leq\widetilde{\Gamma}_{\alpha
}(A;B)_{\rho}+\frac{\alpha}{\alpha-1}\log_{2}\!\left(  \frac{1}{1-\varepsilon
}\right)  .
\end{equation}
Since this upper bound is independent of the number $k$ of rounds, we conclude
the upper bound in~\eqref{eq:LOCC-assisted-up-bnd}.
\end{proof}

\begin{theorem}
\label{thm:asymp-str-conv-bnd}The following upper bound holds%
\begin{equation}
R_{\leftrightarrow}(\rho_{AB}) \leq  \widetilde{R}_{\leftrightarrow}(\rho_{AB})\leq\Gamma(A;B)_{\rho}%
,\label{eq:str-conv-dist-rand-up-bnd}%
\end{equation}
where $\Gamma(A;B)_{\rho}$ is defined from
\eqref{eq:Gamma-quantity-gen-div-def} and~\eqref{eq:q-rel-entr}.
\end{theorem}

\begin{proof}
Consider that, for all $\varepsilon\in(0,1]$ and $\alpha>1$, the following
bound holds from Theorem~\ref{thm:one-shot-bnd}:%
\begin{align}
  \frac{1}{n}R_{\leftrightarrow}^{\varepsilon}(\rho_{AB}^{\otimes
n})
&  \leq\frac{1}{n}\widetilde{\Gamma}_{\alpha}(A^{n};B^{n})_{\rho^{\otimes n}%
}+\frac{\alpha}{n(\alpha-1)}\log_{2}\!\left(  \frac{1}{1-\varepsilon}\right)
\notag \\
&  \leq\widetilde{\Gamma}_{\alpha}(A;B)_{\rho}+\frac{\alpha}{n(\alpha-1)}%
\log_{2}\!\left(  \frac{1}{1-\varepsilon}\right)  ,
\end{align}
where we applied subadditivity of $\widetilde{\Gamma}_{\alpha}$, which follows
from Proposition~\ref{prop:subadd-state-meas}\ and the fact that
$\widetilde{D}_{\alpha}$ is additive. Taking the limit as $n\rightarrow\infty
$, we find that%
\begin{equation}
\limsup_{n\rightarrow\infty}\frac{1}{n}R_{\leftrightarrow}^{\varepsilon}%
(\rho_{AB}^{\otimes n})\leq\widetilde{\Gamma}_{\alpha}(A;B)_{\rho}.
\end{equation}
Since this holds for all $\alpha>1$, we conclude that%
\begin{equation}
\limsup_{n\rightarrow\infty}\frac{1}{n}R_{\leftrightarrow}^{\varepsilon}%
(\rho_{AB}^{\otimes n})\leq\inf_{\alpha>1}\widetilde{\Gamma}_{\alpha
}(A;B)_{\rho}=\Gamma(A;B)_{\rho}.
\end{equation}
The upper bound holds for all $\varepsilon\in\lbrack0,1)$, and so we conclude
the desired 
inequalities in~\eqref{eq:str-conv-dist-rand-up-bnd}.
\end{proof}

\section{Semi-Definite Restrictions}

\label{sec:SDP-relax}

The upper bounds in Theorems~\ref{thm:one-shot-bnd} and
\ref{thm:asymp-str-conv-bnd} are not efficiently computable, because the
$\boldsymbol{\Gamma}$-measure in~\eqref{eq:Gamma-quantity-gen-div-def} involves a bilinear
optimization. One could attempt to use the approach from~\cite{HKT18} to
evaluate $\gamma$, but this does not remove all obstacles to having an efficient method for computing $\boldsymbol{\Gamma}$. %However, instead of doing that, 
Here we consider instead a semi-definite restriction of the $\gamma$-measure in
\eqref{eq:basic-gamma-measure}, which leads to alternative upper
bounds on $R_{\leftrightarrow}^{\varepsilon}(\rho_{AB})$ and $\widetilde
{R}_{\leftrightarrow}(\rho_{AB})$. Indeed, for a given bipartite operator
$\sigma_{AB}$ and state $\rho_{A}$, we define the following semi-definite
restriction of $\gamma(\sigma_{AB})$, taken by fixing $K_{A}=\rho_{A}$ in~\eqref{eq:basic-gamma-measure}:
\begin{align}
 \beta(\sigma_{AB},\rho_{A}) & \coloneqq\inf_{L_{B},V_{AB}\in\text{Herm}%
}\left\{
\begin{array}
[c]{c}%
\operatorname{Tr}[\rho_{A}\otimes L_{B}]:\\
T_{B}(V_{AB}\pm\sigma_{AB})\geq0,\\
\rho_{A}\otimes L_{B}\pm V_{AB}\geq0
\end{array}
\right\} .
\label{eq:beta-measure-states}
\end{align}
It is then clear that $\gamma(\sigma_{AB})\leq\beta(\sigma_{AB},\rho_{A})$ $\forall \rho_A$,
as well as that $\boldsymbol{\Gamma}(\rho_{AB})\leq\boldsymbol{\Upsilon}%
^{A}(\rho_{AB})$, where%
\begin{equation}
\boldsymbol{\Upsilon}^{A}(\rho_{AB})\coloneqq\inf_{\substack{\sigma_{AB}%
\geq0:\\ \beta(\sigma_{AB},\rho_{A})\leq1}}\boldsymbol{D}(\rho_{AB}\Vert
\sigma_{AB}).
\end{equation}
Note that one could also restrict the optimization of $L_{B}$ by fixing
$L_{B}=\rho_{B}$ and obtain the following quantities:
\begin{align}
 \beta(\sigma_{AB},\rho_{B})& \coloneqq\inf_{K_{A},V_{AB}\in\text{Herm}%
}\left\{
\begin{array}
[c]{c}%
\operatorname{Tr}[K_{A}\otimes\rho_{B}]:\\
T_{B}(V_{AB}\pm\sigma_{AB})\geq0,\\
K_{A}\otimes\rho_{B}\pm V_{AB}\geq0
\end{array}
\right\}  ,\notag \\
 \boldsymbol{\Upsilon}^{B}(\rho_{AB})& \coloneqq\inf_{\substack{\sigma_{AB}%
\geq0:\\
\beta(\sigma_{AB},\rho_{B})\leq1}}\boldsymbol{D}(\rho_{AB}\Vert
\sigma_{AB}).
\label{eq:Ups-meas-states}
\end{align}
So then it follows that%
\begin{equation}
\boldsymbol{\Gamma}(\rho_{AB})\leq\min\{\boldsymbol{\Upsilon}^{A}(\rho
_{AB}),\boldsymbol{\Upsilon}^{B}(\rho_{AB})\},\label{eq:restrict-opt-up-bnd}%
\end{equation}
as well as
\begin{align}
R_{\leftrightarrow}^{\varepsilon}(\rho_{AB})  & \leq\min\{\widetilde{\Upsilon
}_{\alpha}^{A}(\rho_{AB}),\widetilde{\Upsilon}_{\alpha}^{B}(\rho_{AB}
)\} +\alpha' \log_{2}\!\left(  \frac{1}{1-\varepsilon}\right)
,\notag \\
\widetilde{R}_{\leftrightarrow}(\rho_{AB})  & \leq\min\{\Upsilon^{A}(A;B)_{\rho
},\Upsilon^{B}(A;B)_{\rho}\},\label{eq:SDP-upper-bnd}
\end{align}
where $\alpha' \equiv \frac{\alpha}{\alpha-1}$,
as an immediate consequence of~\eqref{eq:restrict-opt-up-bnd}\ and
Theorems~\ref{thm:one-shot-bnd} and~\ref{thm:asymp-str-conv-bnd}. To compute the upper bound in the second line of~\eqref{eq:SDP-upper-bnd} efficiently, one can make use of the relative entropy optimization method from~\cite{FF18}. We remark here that the quantities in~\eqref{eq:beta-measure-states}--\eqref{eq:Ups-meas-states} are related to those defined previously in~\cite{WXD18,WFT19,Fang2019a}, and this point is discussed further in Appendix~\ref{app:beta-measures}. 

\begin{figure}
\begin{center}
\includegraphics[
width=8cm]{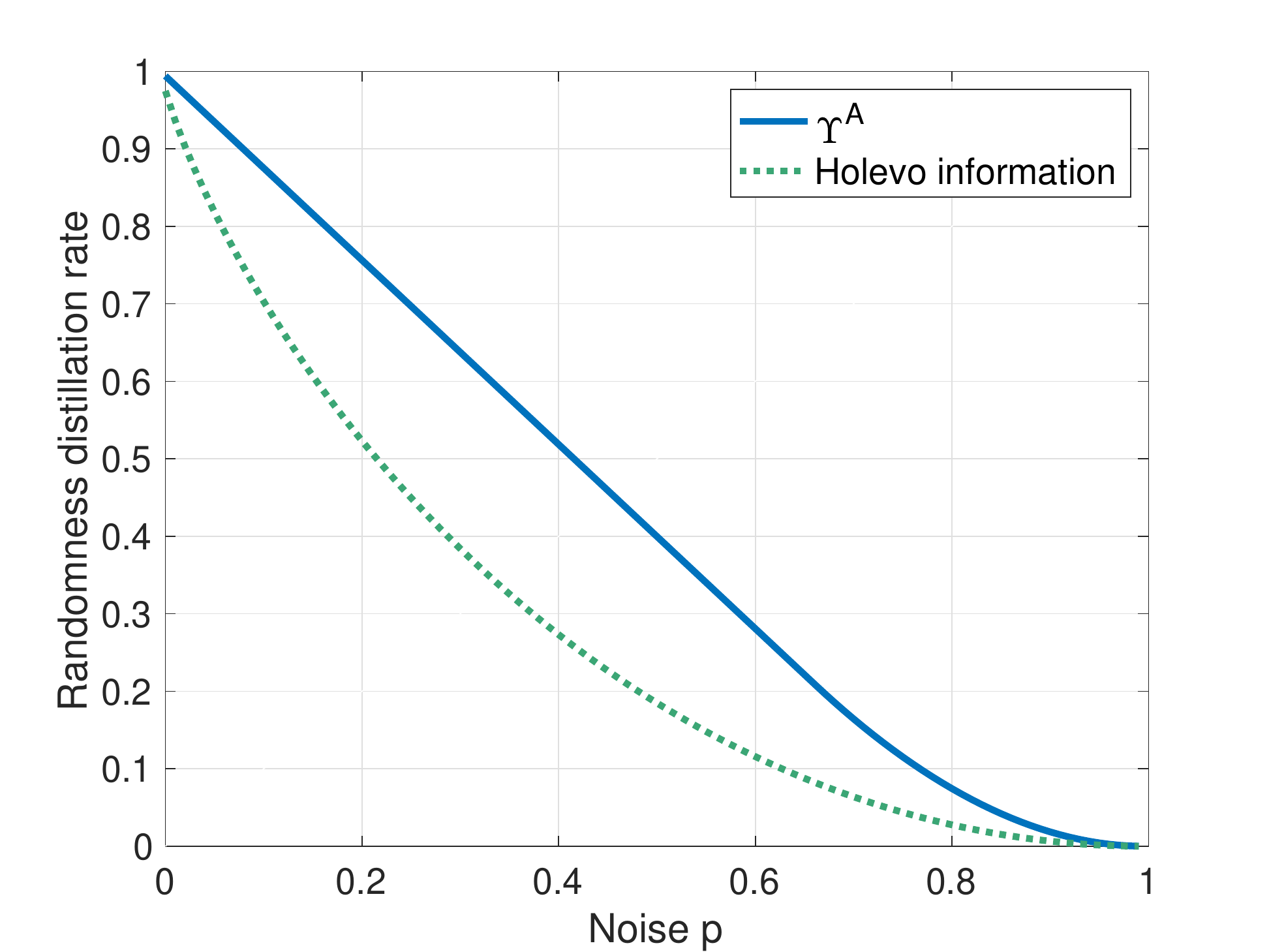}
\end{center}
\caption{Lower and upper bounds on the LOCC-assisted distillable randomness of a qubit-qubit isotropic state, as a function of the parameter $p\in[0,1]$.}
\label{fig:isotrop-plot}%
\end{figure}

As an example, in Figure~\ref{fig:isotrop-plot}, we plot the upper bound in~\eqref{eq:SDP-upper-bnd} for an
isotropic state, defined for $p\in\left[  0,1\right]  $ as $\left(
1-p\right)  \Phi_{AB}^{d}+pI_{AB}/d^{2}$, where $\Phi_{AB}^{d} \coloneqq \frac{1}{d}\sum_{i,j} |i\rangle\!\langle j|_A \otimes |i\rangle\!\langle j|_B $. For comparison, we also plot the
Holevo information lower bound from~\cite{DW04}. The code for generating this figure is available with the arXiv posting of our paper. We note the similarity with Figure~6\ of
\cite{DKQSWW22}, which is for the dynamical case. Clearly, there is a gap between the lower and upper
bounds, and a pertinent question  is to close this gap, just as
is the case for Figure~6 of~\cite{DKQSWW22}.

\section{Conclusion}

\label{sec:conclusion}

In this paper, we returned to the problem of distillable randomness of a
bipartite state, providing a number of upper bounds on this quantity that
are applicable in both the non-asymptotic and asymptotic regimes. To do so, we
introduced a measure of classical correlations contained in a bipartite state.
The main measure that we used to provide an upper bound is not clearly efficiently
computable; however,  we considered a semi-definite restriction
that serves as an upper bound.

Going forward from here, it is open to establish tighter
upper bounds on the distillable randomness. In future work, we plan to apply the recent lower bound on entanglement cost from~\cite[Eq.~(13)]{LR21}, along with the identity in~\cite[Theorem~1]{KW04} that relates entanglement cost and 1W-LOCC distillable randomness.  It is also
open to generalize these methods to the multipartite case (here, see~\cite{SVW05,VC19,SW20}).

{ \textit{Acknowledgements}---We acknowledge discussions with Dawei Ding, Sumeet Khatri, Yihui Quek, and
Peter Shor. LL was partly supported by the Alexander von Humboldt Foundation. BR was supported by the Japan Society for the Promotion of Science (JSPS) KAKENHI Grant No. 21F21015 and the JSPS Postdoctoral Fellowship for Research in Japan.  MMW acknowledges support from NSF\ Grant No.~1907615. }

%%%%%%
%% To balance the columns at the last page of the paper use this
%% command:
%%
%\enlargethispage{-1.2cm}
%%
%% If the balancing should occur in the middle of the references, use
%% the following trigger:
%%
%\IEEEtriggeratref{4}
%%
%% which triggers a \newpage (i.e., new column) just before the given
%% reference number. Note that you need to adapt this if you modify
%% the paper.  The "triggered" command can be changed if desired:
%%
%\IEEEtriggercmd{\enlargethispage{-20cm}}
%%
%%%%%%

%%%%%%
%% References:
%% We recommend the usage of BibTeX:
%%
%\bibliographystyle{IEEEtran}
%\bibliography{definitions,bibliofile}
%%
%% where we here have assume the existence of the files
%% definitions.bib and bibliofile.bib.
%% BibTeX documentation can be obtained at:
%% http://www.ctan.org/tex-archive/biblio/bibtex/contrib/doc/
%%%%%%
\newpage

\bibliographystyle{IEEEtran}
\bibliography{Ref}

%% Or you use manual references (pay attention to consistency and the
%% formatting style!):

\newpage

\appendices

\section{Notation}

\label{app:notation}Here we provide some further notation needed to understand
the proofs in the following appendices. We denote the unnormalized maximally
entangled operator by%
\begin{equation}
\Gamma_{RA}\coloneqq|\Gamma\rangle\!\langle\Gamma|_{RA},\qquad|\Gamma
\rangle_{RA}\coloneqq\sum_{i=0}^{d-1}|i\rangle_{R}|i\rangle_{A},
\end{equation}
where $R\simeq A$ with dimension $d$ and $\{|i\rangle_{R}\}_{i=0}^{d-1}$ and
$\{|i\rangle_{A}\}_{i=0}^{d-1}$ are orthonormal bases. The notation $R\simeq
A$\ means that the systems $R$ and $A$ are isomorphic. The Choi operator of a
quantum channel $\mathcal{N}_{A\rightarrow B}$ is denoted by%
\begin{equation}
\Gamma_{RB}^{\mathcal{N}}\coloneqq\mathcal{N}_{A\rightarrow B}(\Gamma_{RA}).
\end{equation}

\section{Proof of Proposition~\ref{prop:DP-local-chs} (Data Processing under
Local Channels)}

\label{app:data-proc-local-chs}

We prove that $\gamma(\rho_{AB})\geq\gamma(\omega_{A^{\prime}B^{\prime}})$,
which is equivalent to the inequality~\eqref{eq:C-gamma-DP-local-chs} in
Proposition~\ref{prop:DP-local-chs}. Let $K_{A}$, $L_{B}$, and $V_{AB}$ be
arbitrary Hermitian operators satisfying $T_{B}(V_{AB}\pm\rho_{AB})\geq0$ and
$K_{A}\otimes L_{B}\pm V_{AB}\geq0$. Consider that the map $T_{B^{\prime}%
}\circ\mathcal{M}_{B\rightarrow B^{\prime}}\circ T_{B}$ is completely
positive, which follows because%
\begin{align}
&  (T_{B^{\prime}}\circ\mathcal{M}_{B\rightarrow B^{\prime}}\circ
T_{B})(\Gamma_{\hat{B}B})\nonumber\\
&  =(T_{B^{\prime}}\circ\mathcal{M}_{B\rightarrow B^{\prime}}\circ T_{\hat{B}%
})(\Gamma_{\hat{B}B})\\
&  =((T_{\hat{B}}\otimes T_{B^{\prime}})\circ\mathcal{M}_{B\rightarrow
B^{\prime}})(\Gamma_{\hat{B}B})\\
&  =T_{\hat{B}B^{\prime}}(\mathcal{M}_{B\rightarrow B^{\prime}}(\Gamma
_{\hat{B}B}))\geq0.
\end{align}
The last inequality follows because $\mathcal{M}_{B\rightarrow B^{\prime}}$ is
completely positive and $T_{\hat{B}B^{\prime}}$ is a positive map, which in
this case is acting on both systems $\hat{B}B^{\prime}$ and thus preserves
positivity. We also present an alternative proof that $T_{B^{\prime}}%
\circ\mathcal{M}_{B\rightarrow B^{\prime}}\circ T_{B}$ is completely positive
if $\mathcal{M}_{B\rightarrow B^{\prime}}$ is. Consider the following chain of
equalities for a Kraus decomposition of $\mathcal{M}_{B\rightarrow B^{\prime}%
}$ as $\mathcal{M}_{B\rightarrow B^{\prime}}(\cdot)=\sum_{i}M_{i}(\cdot
)M_{i}^{\dag}$:%
\begin{align}
&  (T_{B^{\prime}}\circ\mathcal{M}_{B\rightarrow B^{\prime}}\circ T_{B}%
)(\cdot)\nonumber\\
&  =\sum_{\ell,m}\sum_{j,k}\sum_{i}|\ell\rangle\!\langle m|M_{i}%
|j\rangle\!\langle k|(\cdot)|j\rangle\!\langle k|M_{i}^{\dag}|\ell
\rangle\!\langle m|\\
&  =\sum_{\ell,m}\sum_{j,k}\sum_{i}|\ell\rangle\!\langle k|M_{i}^{\dag}%
|\ell\rangle\!\langle k|(\cdot)|j\rangle\!\langle m|M_{i}|j\rangle\!\langle
m|\\
&  =\sum_{i}\sum_{\ell,k}|\ell\rangle\!\langle k|M_{i}^{\dag}|\ell
\rangle\!\langle k|(\cdot)\sum_{j,m}|j\rangle\!\langle m|M_{i}|j\rangle
\!\langle m|\\
& = \sum_{i} T(M_i^\dag) (\cdot) T(M_i) \\
&  =\sum_{i}E_{i}(\cdot)E_{i}^{\dag},
\end{align}
where $E_{i}\coloneqq T(M_{i}^{\dag})$. Thus, $\{E_{i}\}_{i}$ is a set of Kraus operators
for $T_{B^{\prime}}\circ\mathcal{M}_{B\rightarrow B^{\prime}}\circ T_{B}$,
which implies that this map is completely positive.

Since $\mathcal{N}_{A\rightarrow A^{\prime}}$ and $T_{B^{\prime}}%
\circ\mathcal{M}_{B\rightarrow B^{\prime}}\circ T_{B}$ are completely
positive, it follows that%
\begin{equation}
(\mathcal{N}_{A\rightarrow A^{\prime}}\otimes(T_{B^{\prime}}\circ
\mathcal{M}_{B\rightarrow B^{\prime}}\circ T_{B}))(T_{B}(V_{AB}\pm\rho
_{AB}))\geq0.
\end{equation}
Now consider that%
\begin{align}
&  (\mathcal{N}_{A\rightarrow A^{\prime}}\otimes(T_{B^{\prime}}\circ
\mathcal{M}_{B\rightarrow B^{\prime}}\circ T_{B}))(T_{B}(V_{AB}\pm\rho
_{AB}))\nonumber\\
&  =(\mathcal{N}_{A\rightarrow A^{\prime}}\otimes(T_{B^{\prime}}%
\circ\mathcal{M}_{B\rightarrow B^{\prime}}))(V_{AB}\pm\rho_{AB})\\
&  =T_{B^{\prime}}(\mathcal{N}_{A\rightarrow A^{\prime}}\otimes\mathcal{M}%
_{B\rightarrow B^{\prime}})(V_{AB}\pm\rho_{AB})\\
&  =T_{B^{\prime}}(V_{A^{\prime}B^{\prime}}^{\prime}\pm\omega_{A^{\prime
}B^{\prime}}),
\end{align}
where%
\begin{equation}
V_{A^{\prime}B^{\prime}}^{\prime}\coloneqq(\mathcal{N}_{A\rightarrow
A^{\prime}}\otimes\mathcal{M}_{B\rightarrow B^{\prime}})(V_{AB}).
\end{equation}
So it follows that%
\begin{equation}
T_{B^{\prime}}(V_{A^{\prime}B^{\prime}}^{\prime}\pm\omega_{A^{\prime}%
B^{\prime}})\geq0.
\end{equation}
Since $\mathcal{N}_{A\rightarrow A^{\prime}}$ and $\mathcal{M}_{B\rightarrow
B^{\prime}}$ are completely positive, it follows that%
\begin{equation}
(\mathcal{N}_{A\rightarrow A^{\prime}}\otimes\mathcal{M}_{B\rightarrow
B^{\prime}})(K_{A}\otimes L_{B}\pm V_{AB})\geq0,
\end{equation}
which is equivalent to%
\begin{equation}
K_{A^{\prime}}^{\prime}\otimes L_{B^{\prime}}^{\prime}\pm V_{A^{\prime
}B^{\prime}}^{\prime}\geq0,
\end{equation}
where%
\begin{equation}
K_{A^{\prime}}^{\prime}\coloneqq\mathcal{N}_{A\rightarrow A^{\prime}}%
(K_{A}),\qquad L_{B^{\prime}}^{\prime}\coloneqq\mathcal{M}_{B\rightarrow
B^{\prime}}(L_{B}).
\end{equation}
So we have shown that%
\begin{align}
T_{B}(V_{AB}\pm\rho_{AB})  &  \geq0\quad\Longrightarrow\quad T_{B^{\prime}%
}(V_{A^{\prime}B^{\prime}}^{\prime}\pm\omega_{A^{\prime}B^{\prime}}%
)\geq0,\label{eq:mono-local-pf-1}\\
K_{A}\otimes L_{B}\pm V_{AB}  &  \geq0\quad\Longrightarrow\quad K_{A^{\prime}%
}^{\prime}\otimes L_{B^{\prime}}^{\prime}\pm V_{A^{\prime}B^{\prime}}^{\prime
}\geq0. \label{eq:mono-local-pf-2}%
\end{align}
Also, since $\mathcal{N}_{A\rightarrow A^{\prime}}$ and $\mathcal{M}%
_{B\rightarrow B^{\prime}}$ are trace preserving, it follows that%
\begin{equation}
\operatorname{Tr}[K_{A}\otimes L_{B}]=\operatorname{Tr}[K_{A^{\prime}}%
^{\prime}\otimes L_{B^{\prime}}^{\prime}]. \label{eq:mono-local-pf-3}%
\end{equation}
Putting together~\eqref{eq:mono-local-pf-1}--\eqref{eq:mono-local-pf-3}, we
conclude that%
\begin{multline}
\inf_{K_{A},L_{B},V_{AB}}\left\{
\begin{array}
[c]{c}%
\operatorname{Tr}[K_{A}\otimes L_{B}]:\\
T_{B}(V_{AB}\pm\rho_{AB})\geq0,\\
K_{A}\otimes L_{B}\pm V_{AB}\geq0
\end{array}
\right\} \\
\geq\inf_{K_{A^{\prime}}^{\prime},L_{B^{\prime}}^{\prime},V_{A^{\prime
}B^{\prime}}^{\prime}}\left\{
\begin{array}
[c]{c}%
\operatorname{Tr}[K_{A^{\prime}}^{\prime}\otimes L_{B^{\prime}}^{\prime}]:\\
T_{B}(V_{A^{\prime}B^{\prime}}^{\prime}\pm\omega_{A^{\prime}B^{\prime}}%
)\geq0,\\
K_{A^{\prime}}^{\prime}\otimes L_{B^{\prime}}^{\prime}\pm V_{A^{\prime
}B^{\prime}}^{\prime}\geq0
\end{array}
\right\}  .
\end{multline}
This concludes the proof of~\eqref{eq:C-gamma-DP-local-chs}.

To see~\eqref{eq:bold-gamma-DP-local-chs}, let $\sigma_{AB}$ be an arbitrary
positive semi-definite operator satisfying $\gamma(\sigma_{AB})\leq1$. Then it
follows from~\eqref{eq:C-gamma-DP-local-chs} that $\tau_{A^{\prime}B^{\prime}%
}\coloneqq (\mathcal{N}_{A\rightarrow A^{\prime}}\otimes\mathcal{M}%
_{B\rightarrow B^{\prime}})(\sigma_{AB})$ satisfies%
\begin{equation}
\gamma(\tau_{A^{\prime}B^{\prime}})\leq\gamma(\sigma_{AB})\leq1.
\label{eq:data-proc-loc-chs-gamma}%
\end{equation}
So then, defining $\omega_{A^{\prime}B^{\prime}}\coloneqq (\mathcal{N}%
_{A\rightarrow A^{\prime}}\otimes\mathcal{M}_{B\rightarrow B^{\prime}}%
)(\rho_{AB})$, we find that%
\begin{align}
\boldsymbol{D}(\rho_{AB}\Vert\sigma_{AB})  &  \geq\boldsymbol{D}%
(\omega_{A^{\prime}B^{\prime}}\Vert\tau_{A^{\prime}B^{\prime}})\\
&  \geq\inf_{\substack{\tau_{A^{\prime}B^{\prime}}\geq0,\\\gamma
(\tau_{A^{\prime}B^{\prime}})\leq1}}\boldsymbol{D}(\omega_{A^{\prime}%
B^{\prime}}\Vert\tau_{A^{\prime}B^{\prime}})\\
&  =\boldsymbol{\Gamma}(\omega_{A^{\prime}B^{\prime}}).
\end{align}
The first inequality follows from the data-processing inequality for
$\boldsymbol{D}$, and the second from~\eqref{eq:data-proc-loc-chs-gamma}. The
equality follows from the definition in~\eqref{eq:Gamma-quantity-gen-div-def}.
Since the inequality holds for all $\sigma_{AB}\geq0$ satisfying
$\gamma(\sigma_{AB})\leq1$, we conclude the desired inequality in
\eqref{eq:bold-gamma-DP-local-chs}\ after taking the infimum.

\section{Proof of Proposition~\ref{prop:non-neg-state-meas} (Non-Negativity and Faithfulness)}

\label{app:non-neg}

First, it follows that $C_{\gamma}(A;B)_{\rho}$ takes its minimal value on a
product state, and it is equal to the same value for all product states. This
is because one can transition from an arbitrary state to a product state by
performing local channels that trace out the input and replace with a state.
Indeed, let $\mathcal{R}_{A}^{\sigma}(\cdot)\coloneqq\operatorname{Tr}%
_{A}[\cdot]\sigma_{A}$ and $\mathcal{R}_{B}^{\tau}(\cdot
)\coloneqq\operatorname{Tr}_{B}[\cdot]\tau_{B}$ be local replacer channels.
Then by applying inequality~\eqref{eq:C-gamma-DP-local-chs} in
Proposition~\ref{prop:DP-local-chs}, we conclude that%
\begin{equation}
C_{\gamma}(A;B)_{\rho}\geq C_{\gamma}(A;B)_{\omega}.
\end{equation}
By the same procedure, one can transition from an arbitrary product state to
another arbitrary product state by means of local channels. So the claim
stated above follows. The same argument, but using
\eqref{eq:bold-gamma-DP-local-chs}, implies that $\boldsymbol{\Gamma}%
(\rho_{AB})$ takes its minimal value on product states.

We now prove that this minimal value is zero. By definition, for a product
state $\omega_{AB}=\sigma_{A}\otimes\tau_{B}$,%
\begin{equation}
\gamma(A;B)_{\omega}=\inf_{\substack{K_{A},L_{B},\\V_{AB}\in\text{Herm}%
}}\left\{
\begin{array}
[c]{c}%
\operatorname{Tr}[K_{A}\otimes L_{B}]:\\
T_{B}(V_{AB}\pm\sigma_{A}\otimes\tau_{B})\geq0,\\
K_{A}\otimes L_{B}\pm V_{AB}\geq0
\end{array}
\right\}  .
\end{equation}
Applying some of the constraints, we conclude that%
\begin{align}
T_{B}(\sigma_{A}\otimes\tau_{B})  &  \leq T_{B}(V_{AB}),\\
V_{AB}  &  \leq K_{A}\otimes L_{B}.
\end{align}
Now taking a trace over these constraints, we conclude that%
\begin{align}
1  &  =\operatorname{Tr}[\sigma_{A}\otimes\tau_{B}]=\operatorname{Tr}%
[T_{B}(\sigma_{A}\otimes\tau_{B})]\\
&  \leq\operatorname{Tr}[T_{B}(V_{AB})]=\operatorname{Tr}[V_{AB}%
]\leq\operatorname{Tr}[K_{A}\otimes L_{B}].
\end{align}
So this establishes that $\gamma(A;B)_{\omega}\geq1$ for every product state
$\omega_{AB}$ (and also $\gamma(A;B)_{\rho}\geq1$ for every state $\rho_{AB}$,
by combining the observation in the first paragraph with $\gamma(A;B)_{\omega
}\geq1$ for every product state $\omega_{AB}$).

To see the opposite inequality for a product state $\omega_{AB}$, let us make
the choice $V_{AB}=T_{B}(\sigma_{A}\otimes\tau_{B})$, $K_{A}=\sigma_{A}$, and
$L_{B}=\tau_{B}$. For this choice, all constraints are satisfied, in part
because the partial transpose map is a positive map when acting on a product
state. So it follows that $\gamma(A;B)_{\omega}\leq1$, and combining with what
we previously showed, we conclude that $\gamma(A;B)_{\omega}=1$ for every
product state.

Now we turn to $\boldsymbol{\Gamma}(\rho_{AB})$. Consider that, for an
arbitrary positive semi-definite operator $\sigma_{AB}$, the condition
$\gamma(\sigma_{AB})\leq1$ implies that $\operatorname{Tr}[\sigma_{AB}]\leq1$
because the following holds for arbitrary $V_{AB}$, $K_{A}$, and $L_{B}$
satisfying the constraints in~\eqref{eq:basic-gamma-measure}:%
\begin{align}
\operatorname{Tr}[\sigma_{AB}]  &  =\operatorname{Tr}[T_{B}(\sigma_{AB}%
)]\leq\operatorname{Tr}[T_{B}(V_{AB})]\\
&  =\operatorname{Tr}[V_{AB}]\leq\operatorname{Tr}[K_{A}\otimes L_{B}].
\end{align}
Then taking an infimum over all $V_{AB}$, $K_{A}$, and $L_{B}$ satisfying
\eqref{eq:basic-gamma-measure} and applying the assumption $\gamma(\sigma
_{AB})\leq1$, we conclude that $\operatorname{Tr}[\sigma_{AB}]\leq1$. Now
applying a trace channel to $\boldsymbol{D}(\rho_{AB}\Vert\sigma_{AB})$ and
the non-negative property of a generalized divergence (i.e., $\boldsymbol{D}(1\Vert
c)\geq0,\quad\forall c\in(0,1]$), we conclude that $\boldsymbol{\Gamma}%
(\rho_{AB})\geq0$ for every state $\rho_{AB}$.

If the state of interest is a product state (i.e., $\omega_{AB}=\sigma
_{A}\otimes\tau_{B}$), then it follows that $\gamma(\omega_{AB})\leq1$, as
argued above, so that we can choose $\sigma_{AB}=\omega_{AB}$. With this
choice it follows that $\boldsymbol{\Gamma}(\omega_{AB})\leq\boldsymbol{D}%
(\omega_{AB}\Vert\omega_{AB})=0$, with the latter equality following from the
zero-value property of generalized divergences. Combining with the previous
inequality, we conclude that $\boldsymbol{\Gamma}(\omega_{AB})=0$ for every
product state $\omega_{AB}$.

Let us now turn to the proof that $C_\gamma$ is faithful, i.e.\ that $C_\gamma(\omega_{AB})=0$ implies $\omega_{AB} = \rho_A \otimes \tau_B$ is a product state, and that $\boldsymbol{\Gamma}$ is also faithful provided that the underlying divergence $\boldsymbol{D}$ is lower semicontinuous in its second argument. First, observe that $C_\gamma=\boldsymbol{\Gamma}$ for the particular choice of divergence
\begin{equation}
    \boldsymbol{D}(\rho\|\sigma) = \left\{ \begin{array}{ll} -\log c & \sigma = c\rho , \\ +\infty & \text{otherwise,} \end{array}\right. .
    \label{eq:funny-div}
\end{equation}
This observation is related to the inequality in \eqref{eq:Big-gamma-C-gamma-relation}. Indeed, when performing the optimization for $\boldsymbol{\Gamma}$, consider that choosing $\sigma_{AB}$ to be any operator other than $c \rho_{AB}$ leads to a value of $+\infty$. This then forces $\sigma_{AB}$ to be equal to $c \rho_{AB}$ for some $c> 0$, and the objective function for $\boldsymbol{\Gamma}(\rho_{AB})$ reduces to
\begin{align}
\inf_{c > 0, \gamma(c \rho_{AB})\leq 1}\boldsymbol{D}(\rho_{AB}\|c \rho_{AB}) & = \inf_{c > 0, c \gamma( \rho_{AB})\leq 1 }-\log_2 c \notag \\
& = \inf_{c > 0, c \leq \frac{1}{\gamma( \rho_{AB})} }
-\log_2 c \notag \\
& = \log_2 \gamma( \rho_{AB}),
\end{align}
following from the scaling property of $\boldsymbol{D}$ and the scale invariance of $\gamma( \rho_{AB})$ (Proposition~\ref{prop:scale-invariance}). 
Since the divergence in \eqref{eq:funny-div} happens to be lower semicontinuous in its second argument, it suffices to prove the claim for~$\boldsymbol{\Gamma}$.

Thus, let $\omega_{AB}$ be a state such that $\boldsymbol{\Gamma}(\omega_{AB})=0$. Due to~\eqref{eq:Gamma-quantity-gen-div-def}, for all $\delta>0$ we can find $\sigma_{AB}\geq 0$ such that $\Tr[\sigma_{AB}] \leq \gamma(\sigma_{AB})\leq 1$ and
\begin{equation}
\delta > \boldsymbol{D}(\omega_{AB}\|\sigma_{AB}) \geq \boldsymbol{D}(1\|\Tr [\sigma_{AB}]) = -\log_2 \Tr[\sigma_{AB}]\, ,
\end{equation}
where we also used the data processing inequality and the scaling property for $\boldsymbol{D}$. Combining these two inequalities yields
\begin{equation}
\gamma(\sigma_{AB}) \leq 1 < 2^\delta \Tr \sigma_{AB}\, .
\end{equation}

By definition of $\gamma$, we now find a Hermitian operator $V_{AB}$ and positive semi-definite $K_A$ and $L_B$ such that $T_B\!\left(V_{AB} \pm \sigma_{AB}\right)\geq 0$, $K_A\otimes L_B\pm V_{AB}\geq 0$, and $\Tr\! \left[K_A \otimes L_B\right]\leq 1$. This implies that
\begin{align}
    1 &\geq \Tr[K_A\otimes L_B] \\
    &\geq \Tr [V_{AB}] \\
    &= \Tr [T_B(V_{AB})] \\
    &\geq \Tr [T_B(\sigma_{AB})] \\
    &\geq 2^{-\delta},
\end{align}
and thus in turn that
\begin{align}
    &\left\|K_A\otimes L_B - \sigma_{AB}\right\|_1 \nonumber\\
    & \leq \left\|K_A\otimes L_B - V_{AB} \right\|_1 + \left\|V_{AB} - \sigma_{AB}\right\|_1 \\
    & = \Tr[K_A\otimes L_B - V_{AB}] + \left\|V_{AB} - \sigma_{AB}\right\|_1 \\
    & \leq 1 - 2^{-\delta} + \min\{d_A,d_B\} \left\|T_B(V_{AB} - \sigma_{AB})\right\|_1 \\
    & = 1 - 2^{-\delta} + \min\{d_A,d_B\} \Tr [T_B(V_{AB} - \sigma_{AB})] \\
    & \leq \left(1+\min\{d_A,d_B\}\right) (1-2^{-\delta}) .
\end{align}
In the above calculation, we used the fact that $\|X_{AB}\|_1\leq \min\{d_A,d_B\} \left\|T_B(X_{AB})\right\|_1$ for every Hermitian operator~$X_{AB}$ (see~\cite[proof of Proposition~7]{negativity}), as can be seen immediately by writing a spectral decomposition for $T_B(X_{AB})$ and leveraging the fact that $\left\|T_B(\psi_{AB})\right\|_1\leq \min\{d_A,d_B\}$ for all pure states $\psi_{AB}$~\cite[Proposition~8]{negativity}.

Since $\delta>0$ is arbitrary, we have shown that we can construct a sequence of subnormalized states $\sigma_{AB}(n)$ with the property that (a)~$\lim_{n\to\infty} \boldsymbol{D}\left(\omega_{AB} \| \sigma_{AB}(n) \right) = 0$ and (b)~$\lim_{n\to\infty} \inf_{K_A,L_B} \left\|\sigma_{AB}(n) - K_A\otimes L_B\right\|_1 =0$. Since the set of subnormalized states is compact, up to extracting a subsequence we can assume that $\lim_{n\to\infty} \sigma_{AB}(n) = \sigma_{AB}$; due to~(b) and to the closedness of the set of tensor product operators, we have that $\sigma_{AB} = \rho_A\otimes \tau_B$ is itself a product subnormalized state. But then~(a) together with the lower semicontinuity of $\boldsymbol{D}$ imply that
\begin{align}
    0 &= \lim_{n\to\infty} \boldsymbol{D}(\omega_{AB} \| \sigma_{AB}(n) ) \\
    &\geq \boldsymbol{D}(\omega_{AB} \| \sigma_{AB} ) \\
    &= \boldsymbol{D}(\omega_{AB} \| \rho_A\otimes \tau_B ) .
\end{align}
By faithfulness of $\boldsymbol{D}$, this is only possible if $\omega_{AB} = \rho_A\otimes \tau_B$, which concludes the proof.

\section{Proof of Proposition~\ref{prop:dim-bound-state-measure} (Dimension
Bound)}

\label{app:dim-bnd}The first inequality in~\eqref{eq:dim-bnd} follows from
recalling~\eqref{eq:Big-gamma-C-gamma-relation}. So we prove the second
inequality in~\eqref{eq:dim-bnd}. Let us set $K_{A}=\rho_{A}$, $L_{B}=I_{B}$,
and $V_{AB}=\rho_{A}\otimes I_{B}$. For these choices, we have that%
\begin{equation}
\operatorname{Tr}[K_{A}\otimes L_{B}]=d_{B}.
\end{equation}
We now need to argue that these choices are feasible, i.e., that they satisfy%
\begin{align}
T_{B}(V_{AB}\pm\rho_{AB})  &  \geq0,\label{eq:local-dim-proof-1}\\
K_{A}\otimes L_{B}\pm V_{AB}  &  \geq0. \label{eq:local-dim-proof-2}%
\end{align}
The second set of inequalities is trivially satisfied. For the first, consider
that we need to show that%
\begin{align}
T_{B}(V_{AB}+\rho_{AB})  &  =\rho_{A}\otimes I_{B}+T_{B}(\rho_{AB})\\
&  =\mathcal{P}_{B}^{+}(\rho_{AB})\\
&  \geq0,\\
T_{B}(V_{AB}-\rho_{AB})  &  =\rho_{A}\otimes I_{B}-T_{B}(\rho_{AB})\\
&  =\mathcal{P}_{B}^{-}(\rho_{AB})\\
&  \geq0,
\end{align}
where the linear maps $\mathcal{P}_{B}^{+}$ and $\mathcal{P}_{B}^{-}$ are
defined as%
\begin{align}
\mathcal{P}_{B}^{+}(\cdot)  &  \coloneqq\operatorname{Tr}_{B}[\cdot
]I_{B}+T_{B}(\cdot),\\
\mathcal{P}_{B}^{-}(\cdot)  &  \coloneqq\operatorname{Tr}_{B}[\cdot
]I_{B}-T_{B}(\cdot).
\end{align}
The Choi operators of these maps are given by%
\begin{align}
\mathcal{P}_{B}^{+}(\Gamma_{RB})  &  =I_{RB}+F_{RB}=2\Pi_{RB}^{\mathcal{S}},\\
\mathcal{P}_{B}^{-}(\Gamma_{RB})  &  =I_{RB}-F_{RB}=2\Pi_{RB}^{\mathcal{A}},
\end{align}
where $F_{RB}$ is the unitary swap operator and $\Pi_{RB}^{\mathcal{S}}$ and
$\Pi_{RB}^{\mathcal{A}}$ are the projections onto the symmetric and
antisymmetric subspaces, respectively. Since these operators are positive
semi-definite and they are the Choi operators of $\mathcal{P}_{B}^{+}$ and $\mathcal{P}%
_{B}^{-}$, it follows that $\mathcal{P}_{B}^{+}$
and $\mathcal{P}_{B}^{-}$ are completely positive maps. Thus, the constraints
in~\eqref{eq:local-dim-proof-1} hold, and we conclude the upper bound
$C_{\gamma}(A;B)_{\rho}\leq\log_{2}d_{B}$.

The proof of the other upper bound is similar, and we show it for
completeness. Pick $K_{A}=I_{A}$, $L_{B}=\rho_{B}$, and $V_{AB}=I_{A}%
\otimes\rho_{B}$. For these choices, we have that%
\begin{equation}
\operatorname{Tr}[K_{A}\otimes L_{B}]=d_{A}.
\end{equation}
We need to argue that these choices are feasible. The constraint in
\eqref{eq:local-dim-proof-2} holds trivially. The first constraints in
\eqref{eq:local-dim-proof-1}\ become%
\begin{align}
T_{B}(V_{AB}+\rho_{AB})  &  =I_{A}\otimes T_{B}(\rho_{B})+T_{B}(\rho_{AB}%
)\geq0,\\
T_{B}(V_{AB}-\rho_{AB})  &  =I_{A}\otimes T_{B}(\rho_{B})-T_{B}(\rho_{AB}%
)\geq0.
\end{align}
The inequalities above are equivalent to the following inequalities because
they are related by taking a full transpose:%
\begin{align}
I_{A}\otimes\rho_{B}+T_{A}(\rho_{AB})  &  \geq0,\\
I_{A}\otimes\rho_{B}-T_{A}(\rho_{AB})  &  \geq0.
\end{align}
Then we conclude that the constraints are satisfied because $\mathcal{P}%
_{A}^{+}(\rho_{AB})=I_{A}\otimes\rho_{B}+T_{A}(\rho_{AB})$ and $\mathcal{P}%
_{A}^{-}(\rho_{AB})=I_{A}\otimes\rho_{B}-T_{A}(\rho_{AB})$, and we already
proved that $\mathcal{P}_{A}^{+}$ and $\mathcal{P}_{A}^{-}$ are completely positive.

\section{Proof of Proposition~\ref{prop:classical-comm-bnd} (Classical
Communication Bound)}

\label{app:classical-comm-bnd}

Let $K_{A}$, $L_{BX}$, and $V_{ABX}$ be arbitrary operators for the
optimization problem for $C_{\gamma}(A;BX)_{\rho}$, which satisfy
\begin{align}
T_{BX}(V_{ABX}\pm\rho_{ABX})  &  \geq0,\\
K_{A}\otimes L_{BX}\pm V_{ABX}  &  \geq0.
\end{align}
Pick%
\begin{align}
K_{AX}^{\prime}  &  \coloneqq K_{A}\otimes I_{X},\\
L_{B}^{\prime}  &  \coloneqq\operatorname{Tr}_{X}[L_{BX}],\\
V_{ABX}^{\prime}  &  \coloneqq\overline{\Delta}_{X}(V_{ABX}),
\end{align}
where $\overline{\Delta}_{X}(\cdot)\coloneqq \sum_{x}|x\rangle\!\langle
x|(\cdot)|x\rangle\!\langle x|$ is the completely dephasing channel. Then we
find that%
\begin{align}
\operatorname{Tr}[K_{AX}^{\prime}\otimes L_{B}^{\prime}]  &
=\operatorname{Tr}[K_{A}\otimes I_{X}\otimes\operatorname{Tr}_{X}%
[L_{BX}]]\label{eq:imp-trace-eq-cl-comm-1}\\
&  =\operatorname{Tr}[K_{A}\otimes L_{BX}]\operatorname{Tr}[I_{X}]\\
&  =\operatorname{Tr}[K_{A}\otimes L_{BX}]d_{X}.
\label{eq:imp-trace-eq-cl-comm-3}%
\end{align}
We then need to show that%
\begin{align}
T_{B}(V_{ABX}^{\prime}\pm\rho_{ABX})  &  \geq0,\\
K_{AX}^{\prime}\otimes L_{B}^{\prime}\pm V_{ABX}^{\prime}  &  \geq0.
\end{align}
Since $\overline{\Delta}_{X}$ is a completely positive map, we see that
\begin{align}
&\hspace{10ex} K_{A}\otimes L_{BX}\pm V_{ABX} \geq0\\
&\,\Longrightarrow\qquad K_{A}\otimes\overline{\Delta}_{X}(L_{BX})\pm\overline
{\Delta}_{X}(V_{ABX})  \geq0\\
&\Longleftrightarrow\hspace{9.6ex} K_{A}\otimes\overline{\Delta}_{X}(L_{BX})\pm
V_{ABX}^{\prime}  \geq0.
\end{align}
We also have that $\overline{\Delta}_{X}(L_{BX})\leq L_{B}\otimes I_{X}$,
which follows because%
\begin{align}
\overline{\Delta}_{X}(L_{BX})  &  =\sum_{x}L_{B}^{x}\otimes|x\rangle\!\langle
x|_{X}\\
&  \leq\sum_{x}L_{B}^{x}\otimes I_{X}\\
&  =L_{B}\otimes I_{X},
\end{align}
where%
\begin{equation}
L_{B}^{x}\coloneqq {\vphantom{)}}_X\langle x|L_{BX}|x{\rangle}_{X}.
\end{equation}
The inequality $\overline{\Delta}_{X}(L_{BX})\leq L_{B}\otimes I_{X}$ implies
that%
\begin{align}
K_{A}\otimes L_{B}\otimes I_{X}\pm V_{ABX}^{\prime}  &  \geq0\\
\Longleftrightarrow\qquad K_{AX}^{\prime}\otimes L_{B}^{\prime}\pm V_{ABX}%
^{\prime}  &  \geq0.
\end{align}
Now consider that%
\begin{align}
T_{BX}(V_{ABX}\pm\rho_{ABX})  &  \geq0\\
\Longleftrightarrow\hspace{8.2ex}\overline{\Delta}_{X}(T_{BX}(V_{ABX}\pm\rho_{ABX}))  &
\geq0\\
\Longleftrightarrow\qquad T_{B}(\overline{\Delta}_{X}(V_{ABX})\pm\overline{\Delta
}_{X}(\rho_{ABX}))  &  \geq0\\
\Longleftrightarrow\hspace{15ex} T_{B}(V_{ABX}^{\prime}\pm\rho_{ABX})  &  \geq0.
\end{align}
In the above, we applied the equality $\overline{\Delta}_{X}\circ
T_{X}=\overline{\Delta}_{X}$. Since $K_{AX}^{\prime}$, $L_{B}^{\prime}$, and
$V_{ABX}^{\prime}$ are specific choices that satisfy the constraints for
$C_{\gamma}(AX;B)_{\rho}$, we conclude the desired inequality after minimizing
over all such operators and noticing that the objective function for
$C_{\gamma}(AX;B)_{\rho}$ is $\operatorname{Tr}[K_{AX}^{\prime}\otimes
L_{B}^{\prime}]$, which satisfies
\eqref{eq:imp-trace-eq-cl-comm-1}--\eqref{eq:imp-trace-eq-cl-comm-3}. This
establishes~\eqref{eq:c_gam-comm-bnd}.

To see~\eqref{eq:big_gam-comm-bnd}, consider that%
\begin{align}
&  \log_{2}d_{X}+\boldsymbol{\Gamma}(A;BX)_{\rho}\nonumber\\
&  =\log_{2}d_{X}+\inf_{\substack{\sigma_{ABX}\geq0,\\\gamma(A;BX)_{\sigma
}\leq1}}\boldsymbol{D}(\rho_{ABX}\Vert\sigma_{ABX})\\
&  =\inf_{\substack{\sigma_{ABX}\geq0,\\\gamma(A;BX)_{\sigma}\leq
1}}\boldsymbol{D}(\rho_{ABX}\Vert\sigma_{ABX}/d_{X})\\
&  =\inf_{\substack{\sigma_{ABX}^{\prime}\geq0,\\ \gamma(A;BX)_{d_X \sigma^{\prime
}}\leq1}} \boldsymbol{D}(\rho_{ABX}\Vert\sigma_{ABX}^{\prime})\\
&  =\inf_{\substack{\sigma_{ABX}^{\prime}\geq0,\\d_{X}\,\gamma(A;BX)_{\sigma
^{\prime}}\leq1}}\boldsymbol{D}(\rho_{ABX}\Vert\sigma_{ABX}^{\prime})\\
&  \geq \inf_{\substack{\sigma_{ABX}^{\prime}\geq0,\\\gamma(AX;B)_{\sigma^{\prime
}}\leq1}}\boldsymbol{D}(\rho_{ABX}\Vert\sigma_{ABX}^{\prime})\\
&  =\boldsymbol{\Gamma}(AX;B)_{\rho}.
\end{align}
The fourth equality follows from Proposition~\ref{prop:scale-invariance}, and
the subsequent inequality from~\eqref{eq:c_gam-comm-bnd}.

\section{Proof of Proposition~\ref{prop:subadd-state-meas}\ (Subadditivity)}

\label{app:subadd-state-meas}Let $K_{A_{1}}^{1}$, $L_{B_{1}}^{1}$, and
$V_{A_{1}B_{1}}^{1}$ be feasible choices for $C_{\gamma}(A_{1};B_{1})_{\sigma
}$ (satisfying the constraints in~\eqref{eq:basic-gamma-measure}), and let
$K_{A_{2}}^{2}$, $L_{B_{2}}^{2}$, and $V_{A_{2}B_{2}}^{2}$ be feasible choices
for $C_{\gamma}(A_{2};B_{2})_{\tau}$. Then it follows that $K_{A_{1}}%
^{1}\otimes K_{A_{2}}^{2}$, $L_{B_{1}}^{1}\otimes L_{B_{2}}^{2}$, and
$V_{A_{1}B_{1}}^{1}\otimes V_{A_{2}B_{2}}^{2}$ are feasible choices
$C_{\gamma}(A_{1}A_{2};B_{1}B_{2})_{\rho}$. This latter statement is a
consequence of the general fact that if $A$, $B$, $C$, and $D$ are Hermitian
operators satisfying $A\pm B\geq0$ and $C\pm D\geq0$, then $A\otimes C\pm
B\otimes D\geq0$. To see this, consider that the original four operator
inequalities imply the four operator inequalities $\left(  A\pm B\right)
\otimes\left(  C\pm D\right)  \geq0$, and then summing these four different
operator inequalities in various ways leads to $A\otimes C\pm B\otimes D\geq0$.

To see the inequality in~\eqref{eq:subadd-big-gamma}, let $\omega_{A_{1}B_{1}%
}\geq0$ satisfy $\gamma(A_{1};B_{1})_{\omega}\leq1$ and let $\theta
_{A_{2}B_{2}}\geq0$ satisfy $\gamma(A_{2};B_{2})_{\theta}\leq1$. Then, by
applying~\eqref{eq:subadd-c-gamma}, we have that $\gamma(A_{1}A_{2};B_{1}%
B_{2})_{\omega \otimes \theta}\leq\gamma(A_{1};B_{1})_{\omega}\cdot\gamma(A_{2};B_{2})_{\theta
}\leq1$, so that%
\begin{align}
&  \boldsymbol{\Gamma}(A_{1}A_{2};B_{1}B_{2})_{\rho}\\
&  =\inf_{\substack{\zeta_{A_{1}A_{2}B_{1}B_{2}}\geq0,\\\gamma(\zeta
_{A_{1}A_{2}B_{1}B_{2}})\leq1}}\boldsymbol{D}(\rho_{A_{1}A_{2}B_{1}B_{2}}%
\Vert\zeta_{A_{1}A_{2}B_{1}B_{2}})\\
&  =\inf_{\substack{\zeta_{A_{1}A_{2}B_{1}B_{2}}\geq0,\\\gamma(\zeta
_{A_{1}A_{2}B_{1}B_{2}})\leq1}}\boldsymbol{D}(\sigma_{A_{1}B_{1}}\otimes
\tau_{A_{2}B_{2}}\Vert\zeta_{A_{1}A_{2}B_{1}B_{2}})\\
&  \leq\boldsymbol{D}(\sigma_{A_{1}B_{1}}\otimes\tau_{A_{2}B_{2}}\Vert
\omega_{A_{1}B_{1}}\otimes\theta_{A_{2}B_{2}})\\
&  \leq\boldsymbol{D}(\sigma_{A_{1}B_{1}}\Vert\omega_{A_{1}B_{1}%
})+\boldsymbol{D}(\tau_{A_{2}B_{2}}\Vert\theta_{A_{2}B_{2}}).
\end{align}
The last inequality follows from the hypothesis that $\boldsymbol{D}$ is
subadditive. Since the inequality holds for all $\omega_{A_{1}B_{1}}$ and
$\theta_{A_{2}B_{2}}$ satisfying the constraints for $\boldsymbol{\Gamma
}(A_{1};B_{1})_{\sigma}$ and $\boldsymbol{\Gamma}(A_{2};B_{2})_{\tau}$,
respectively, we conclude the desired inequality in~\eqref{eq:subadd-big-gamma}.

\section{Proof of Lemma~\ref{lem:fid-bnd-free-ops}}

\label{app:pf-fid-bnd-free-ops}Let $\sigma_{AB}$ satisfy $\sigma_{AB}\geq0$
and let $K_{A}$, $L_{B}$, and $V_{AB}$ be Hermitian operators satisfying the
constraints for $\gamma(\sigma_{AB})$. Consider that%
\begin{equation}
F\Big(\overline{\Phi}_{AB}^{d},\sigma_{AB}\Big)\leq F\Big(\overline{\Phi}_{AB}^{d},\left(
\overline{\Delta}_{A}\otimes\overline{\Delta}_{B}\right)  \left(  \sigma
_{AB}\right)  \Big),
\end{equation}
where $\overline{\Delta}$ is a local dephasing channel, defined as%
\begin{equation}
\overline{\Delta}(\cdot)\coloneqq\sum_{m=0}^{d-1}|m\rangle\!\langle
m|(\cdot)|m\rangle\!\langle m|.
\end{equation}
We then find that
\begin{align}
&  \left(  \overline{\Delta}_{A}\otimes\overline{\Delta}_{B}\right)  \left(
\sigma_{AB}\right) \nonumber\\
&  =\sum_{\ell,m=0}^{d-1} |\ell m\rangle\!\langle \ell m|_{AB}\sigma_{AB}|\ell m\rangle\!\langle \ell m|_{AB} \\
%|\ell\rangle\!\langle\ell|_{A}\otimes|m\rangle\!\langle m|_{B}\sigma_{AB}|\ell\rangle\!\langle\ell|_{A}\otimes |m\rangle\!\langle m|_{B}} \\
&  =\sum_{\ell,m=0}^{d-1}\operatorname{Tr}\left[|\ell\rangle\!\langle\ell
|_{A}\otimes|m\rangle\!\langle m|_{B}\,\sigma_{AB}\right] |\ell\rangle\!\langle
\ell|_{A}\otimes|m\rangle\!\langle m|_{B}.
\end{align}
Then the fidelity $F\Big(\overline{\Phi}_{AB}^{d},\left(  \overline{\Delta}%
_{A}\otimes\overline{\Delta}_{B}\right)  \left(  \sigma_{AB}\right)  \Big)$
becomes the classical fidelity, i.e.,
\begin{multline}
F\Big(\overline{\Phi}_{AB}^{d},\left(  \overline{\Delta}_{A}\otimes\overline
{\Delta}_{B}\right)  \left(  \sigma_{AB}\right)  \Big)\\
=\left[  \sum_{m=0}^{d-1}\sqrt{\frac{1}{d}\operatorname{Tr}[|m\rangle\!\langle
m|\otimes|m\rangle\!\langle m|\sigma_{AB}]}\right]  ^{2}.
\end{multline}
Then we use the fact that%
\begin{align}
&  \operatorname{Tr}[|m\rangle\!\langle m|_{A}\otimes|m\rangle\!\langle
m|_{B}\sigma_{AB}]\nonumber\\
&  =\operatorname{Tr}[|m\rangle\!\langle m|_{A}\otimes|m\rangle\!\langle
m|_{B}T_{B}(\sigma_{AB})]\\
&  \leq\operatorname{Tr}[|m\rangle\!\langle m|_{A}\otimes|m\rangle\!\langle
m|_{B}T_{B}(V_{AB})]\\
&  =\operatorname{Tr}[|m\rangle\!\langle m|_{A}\otimes|m\rangle\!\langle
m|_{B}V_{AB}]\\
&  \leq\operatorname{Tr}[|m\rangle\!\langle m|_{A}\otimes|m\rangle\!\langle
m|_{B}K_{A}\otimes L_{B}]\\
&  =\langle m|K_{A}|m\rangle\!\langle m|L_{B}|m\rangle.
\end{align}
so that the classical fidelity is less than%
\begin{align}
&  \left[  \sum_{m}\sqrt{\frac{1}{d}\langle m|K_{A}|m\rangle\!\langle
m|L_{B}|m\rangle}\right]  ^{2}\nonumber\\
&  \leq\frac{1}{d}\sum_{m}\langle m|K_{A}|m\rangle\sum_{m}\langle
m|L_{B}|m\rangle\\
&  =\frac{1}{d}\operatorname{Tr}[K_{A}]\operatorname{Tr}[L_{B}]\\
&  =\frac{1}{d}\operatorname{Tr}[K_{A}\otimes L_{B}].
\end{align}
We applied the Cauchy--Schwarz inequality. Now %applying
leveraging the assumption that $\gamma(\sigma_{AB})\leq1$ and taking an infimum over all $K_{A}$, $L_{B}$,
and $V_{AB}$ satisfying the constraints for $\gamma(\sigma_{AB})$, we conclude
that%
\begin{equation}
F\Big(\overline{\Phi}_{AB}^{d},\sigma_{AB}\Big)\leq\frac{1}{d}.
\end{equation}
We conclude the statement of the lemma because we have proven that this
inequality holds for all $\sigma_{AB}$ satisfying $\sigma_{AB}\geq0$ and
$\gamma(\sigma_{AB})\leq1$.

\section{Proof of Proposition~\ref{prop:key-dim-bnd-converse}}

\label{app:pseudo-cont-bnd}Recall the following inequality from~\cite[Lemma~1]%
{Wang2019states}:%
\begin{equation}
\widetilde{D}_{\alpha}(\rho_{0}\Vert\sigma)-\widetilde{D}_{\beta}(\rho
_{1}\Vert\sigma)\geq\frac{\alpha}{\alpha-1}\log_{2}F(\rho_{0},\rho_{1}),
\end{equation}
where $\rho_{0}$ and $\rho_{1}$ are states, $\sigma$ is a positive
semi-definite operator, $\alpha>1$, and $\beta$ satisfies $\frac{1}{2\alpha
}+\frac{1}{2\beta}=1$, so that $\beta\in(1/2,1)$. By the same argument given
in~\cite[Lemma~14]{EW22}, this implies that%
\begin{align}
\widetilde{\Gamma}_{\alpha}(\omega_{AB})-\widetilde{\Gamma}_{\beta}%
\Big(\overline{\Phi}_{AB}^{d}\Big) &  \geq\frac{\alpha}{\alpha-1}\log_{2}%
F\Big(\overline{\Phi}_{AB}^{d},\omega_{AB}\Big)\\
&  \geq\frac{\alpha}{\alpha-1}\log_{2}(1-\varepsilon).
\end{align}
We can rewrite this as%
\begin{align}
\widetilde{\Gamma}_{\alpha}(\omega_{AB})+\frac{\alpha}{\alpha-1}\log
_{2}\left(  \frac{1}{1-\varepsilon}\right)   &  \geq\widetilde{\Gamma}_{\beta
}\Big(\overline{\Phi}_{AB}^{d}\Big)\\
&  \geq\widetilde{\Gamma}_{1/2}\Big(\overline{\Phi}_{AB}^{d}\Big)\\
&  \geq\log_{2}d,
\end{align}
where the 
second inequality follows from $\beta$-monotonicity of the sandwiched
Renyi relative entropy (see~\cite[Prop.~4.29]{KW20book}), and the 
third from Lemma~\ref{lem:fid-bnd-free-ops}.

\section{Relation to $\beta$- and $\boldsymbol{\Upsilon}$-Measures of~\cite{WXD18,WFT19,Fang2019a}}

\label{app:beta-measures}

In this appendix, we discuss the relation of the $\gamma$, $C_\gamma$, and~$\boldsymbol{\Gamma}$ measures introduced in the main text, with the $\beta$, $C_\beta$, and~$\boldsymbol{\Upsilon}$ measures from~\cite{WXD18,WFT19,Fang2019a}. 

Indeed, by fixing the operator $K_{A}$ in~\eqref{eq:basic-gamma-measure} to be the marginal operator $\sigma_{A}$ (where $\sigma_{AB}$ is a bipartite positive semi-definite operator), we
arrive at the following quantity:%
\begin{equation}
\beta(\sigma_{AB})\coloneqq\inf_{L_{B},V_{AB}\in\text{Herm}}\left\{
\begin{array}
[c]{c}%
\operatorname{Tr}[\sigma_{A}\otimes L_{B}]:\\
T_{B}(V_{AB}\pm\sigma_{AB})\geq0,\\
\sigma_{A}\otimes L_{B}\pm V_{AB}\geq0
\end{array}
\right\}  .
\end{equation}
This is precisely the static version of the $\beta$ measure from~\cite[Eq.~(45)]{WXD18}. We also define
\begin{equation}
C_{\beta}(\sigma_{AB})    \coloneqq C_{\beta}(A;B)_{\sigma}\coloneqq\log
_{2}\beta(\sigma_{AB}).    
\end{equation}
For a bipartite state $\rho_{AB}$, we then define
\begin{equation}
\boldsymbol{\Upsilon}(\rho_{AB})   \coloneqq\inf_{\substack{\sigma_{AB}%
\geq0:\\\beta(\sigma_{AB})\leq1}}\boldsymbol{D}(\rho_{AB}\Vert\sigma_{AB}),  
\end{equation}
which is the static version of the $\boldsymbol{\Upsilon}$ measure from~\cite[Eq.~(49)]{WFT19}. 
The following inequalities clearly hold, by applying definitions:
\begin{equation}
\gamma(\sigma_{AB})\leq\beta(\sigma_{AB}),\qquad\boldsymbol{\Gamma}(\rho
_{AB})\leq\boldsymbol{\Upsilon}(\rho_{AB}).
\end{equation}

Thus, $\beta$, $C_{\beta}$, and $\boldsymbol{\Upsilon}$ can be
understood as static counterparts of the corresponding measures of classical
correlations of a channel, from~\cite{WXD18,WFT19,Fang2019a}. We
wrote them down explicitly above to make the connection with prior literature.
However, in the main text, we focused exclusively on $\gamma$, $C_{\gamma}$, and
$\boldsymbol{\Gamma}$ because these quantities lead to tighter upper bounds on
the distillable randomness. The semi-definite restrictions of
these quantities in Section~\ref{sec:SDP-relax} are closely related as well with $\beta$, $C_{\beta}$, and $\boldsymbol{\Upsilon}$, and as discussed there, they lead to computationally efficient upper bounds on distillable randomness.

We finally note that the static $\beta$- and $\boldsymbol{\Upsilon}$-measures for a bipartite state and the dynamic ones for point-to-point channels are further generalized by the corresponding measures for a bipartite channel, as proposed in~\cite{DKQSWW22}. As such, the $\gamma$, $C_\gamma$, and $\boldsymbol{\Gamma}$ measures can be generalized to bipartite channels, as considered in~\cite{W20unpub}, and will be the subject of a future publication.

\end{document}